\documentclass{article}
\usepackage{fullpage}
\usepackage{amsthm}
\usepackage{amsmath}
\usepackage{amssymb}
\usepackage{authblk}
\usepackage{t1enc}
\usepackage[utf8]{inputenc}
\usepackage{colonequals}
\usepackage{float}
\usepackage{tikz}
\usetikzlibrary{calc,arrows.meta, decorations.shapes, shapes.geometric}
\usetikzlibrary{shapes,positioning,decorations.markings}
\usepackage{booktabs}
\usepackage{array}
\usepackage{multicol}
\usepackage{multirow}
\usepackage{caption}
\usepackage{enumerate}
\usepackage{enumitem}
\usepackage{xspace}
\usepackage{titling}
\thanksmarkseries{arabic}

\usepackage[hang,flushmargin]{footmisc}

\makeatletter
 \g@addto@macro \normalsize{\setlength\abovedisplayskip{5pt plus 0pt minus 0pt} \setlength\belowdisplayskip{5pt plus 0pt minus 0pt}}
\makeatother

\newtheorem{thm}{Theorem}[section]
\newtheorem{claim}[thm]{Claim}
\newtheorem{cor}[thm]{Corollary}

\def\R{\mathbb R}
\def\Z{\mathbb Z}

\DeclareFontShape{T1}{cmr}{m}{scit}{<->ssub*cmr/m/sc}{}

\usepackage[unicode,colorlinks=true,citecolor=green!40!black,linkcolor=red!20!black,urlcolor=blue!40!black,filecolor=cyan!30!black]{hyperref}

\newcommand*{\ExDirSPEx}{{\scshape Exact-Si\-mul\-ta\-ne\-ous-Dipaths-Ex\-is\-tence}\@\xspace}
\newcommand*{\ExDirSPMinCost}{{\scshape Exact-Si\-mul\-ta\-ne\-ous-Dipaths}\@\xspace}

\newcommand*{\ExUndSPEx}{{\scshape Exact-Si\-mul\-ta\-ne\-ous-Paths-Ex\-is\-tence}\@\xspace}
\newcommand*{\ExUndSPMinCost}{{\scshape Exact-Si\-mul\-ta\-ne\-ous-Paths}\@\xspace}

\newcommand*{\SupDirSPMinCost}{{\scshape Superset-Si\-mul\-ta\-ne\-ous-Di\-paths}\@\xspace}

\newcommand*{\SupUndSPMinCost}{{\scshape Superset-Si\-mul\-ta\-ne\-ous-Paths}\@\xspace}

\pgfdeclarelayer{background}
\pgfdeclarelayer{foreground}
\pgfsetlayers{background,main,foreground}

\tikzset{
    my box/.style = {
        , line cap = round
        , line join = round
    }
}

\newcommand{\highlight}[3]{
  \path [my box, line width = #1, draw = #2, transparency group, opacity=1] #3;
}

\title{Simultaneous Network Design with Restricted Link Usage}

\author{Naonori Kakimura\thanks{Keio University, E-mail: \texttt{kakimura@math.keio.ac.jp}}
\and P\'eter Madarasi\thanks{HUN-REN R\'{e}nyi Alfr\'{e}d R\'{e}nyi Institute of Mathematics, and Department of Operations Research, E{\"o}tv{\"o}s Lor{\'a}nd University, Budapest, Hungary. E-mail: \texttt{madarasip@staff.elte.hu}}
\and Jannik Matuschke\thanks{Research Center for Operations Management, KU Leuven, Belgium. E-mail: \texttt{jannik.matuschke@kuleuven.be}}
\and Kitti Varga\thanks{HUN-REN--ELTE Egerv{\'a}ry Research Group on Combinatorial Optimization, and Department of Computer Science and Information Theory, Faculty of Electrical Engineering and Informatics, Budapest University of Technology and Economics, Budapest, Hungary. E-mail: \texttt{vkitti@cs.bme.hu}}}

\date{}

\begin{document}

\setlength\abovedisplayskip{1pt plus 2pt minus 3pt}
\setlength\belowdisplayskip{3pt plus 2pt minus 3pt}

\maketitle

\begin{abstract}
  Given a digraph with two terminal vertices $s$ and $t$ as well as a conservative cost function and several not necessarily disjoint color classes on its arc set, our goal is to find a minimum-cost subset of the arcs such that its intersection with each color class contains an $s$-$t$ dipath.
  Problems of this type arise naturally in multi-commodity network design settings where each commodity is restricted to use links of its own color only.

  We study several variants of the problem, deriving strong hardness results even for restricted cases, but we also identify cases that can be solved in polynomial time.
  The latter ones include the cases where the color classes form a laminar family, or where the underlying digraph is acyclic and the number of color classes is constant.
  We also present an FPT algorithm for the general case parameterized by the number of multi-colored arcs.

\end{abstract}
\bigskip

\section{Introduction}

The study of connectivity problems in graphs is a cornerstone of combinatorial optimization, supporting decision makers in the cost-efficient design of network infrastructures in diverse areas such as transportation, communication, or energy transmission.
In the fundamental shortest path problem, two terminal vertices have to be connected by a path of minimum cost.
In typical applications, it is natural to consider different commodities using the resulting network infrastructure, with each link of the network suited to support only some subset of commodities.
For example, a traffic network may be used by different vehicle types such as cars, buses, trucks, or bicycles, but not every link in the network may be suitable for use by each vehicle type.

In this paper, we study a new, fundamental network optimization problem that captures this type of usage restrictions.
In this simultaneous path problem, we are given a directed graph~(\emph{digraph} for short) with two special vertices $s$ and $t$ and a conservative cost function (i.e., no negative-cost directed cycles are allowed) on the arc set, as well as $k$ not necessarily disjoint subsets of the arcs --- called \emph{color classes} --- covering the arc set.
Our goal is to find a minimum-cost subset of the arcs such that its intersection with each of the $k$ color classes contains an $s$-$t$ directed path~(\emph{dipath} for short).
In a stricter variant of the problem, we require the intersection of the arc set with each color class to be an $s$-$t$ dipath rather than containing one.
We refer to these two variants as the ``superset'' and the ``exact'' variant, respectively.
Allowing for either directed or undirected graph, we obtain four variants of the problem, which we formally state below.

\medskip
\noindent {\bf \ExDirSPMinCost} \\
\textit{Input:} a positive integer $k$, a digraph~$D=(V,A)$ with $A = A_1 \cup \cdots \cup A_k$, two distinct vertices $s, t \in V$, and a conservative cost function $c \colon A \to \mathbb{R}$.\\
\textit{Problem:} find $A' \subseteq A$ of minimum cost such that $A' \cap A_i$ is an $s$-$t$ dipath for all $i \in \{ 1, \ldots, k \}$.
\medskip

\noindent {\bf \ExUndSPMinCost} \\
\textit{Input:} a positive integer $k$, a graph~$G=(V,E)$ with $E = E_1 \cup \cdots \cup E_k$, two distinct vertices $s, t \in V$, and a nonnegative cost function $c \colon A \to \mathbb{R}_+$.\\
\textit{Problem:} find $E' \subseteq E$ of minimum cost such that $E' \cap E_i$ is an $s$-$t$ path for all $i \in \{ 1, \ldots, k \}$.
\medskip

\noindent {\bf \SupDirSPMinCost} \\
\textit{Input:} a positive integer $k$, a digraph~$D=(V,A)$ with $A = A_1 \cup \cdots \cup A_k$, two distinct vertices $s, t \in V$, and a conservative cost function $c \colon A \to \mathbb{R}$.\\
\textit{Problem:} find $A' \subseteq A$ of minimum cost such that $A' \cap A_i$ contains an $s$-$t$ dipath for all $i \in \{ 1, \ldots, k \}$.
\medskip

\noindent {\bf \SupUndSPMinCost} \\
\textit{Input:} a positive integer $k$, a graph~$G=(V,E)$ with $E = E_1 \cup \cdots \cup E_k$, two distinct vertices $s, t \in V$, and a nonnegative cost function $c \colon A \to \mathbb{R}_+$.\\
\textit{Problem:} find $E' \subseteq E$ of minimum cost such that $E' \cap E_i$ contains an $s$-$t$ path for all $i \in \{ 1, \ldots, k \}$.
\bigskip

For some examples on the above problems and their optimal solutions, see Figure~\ref{fig:example}.

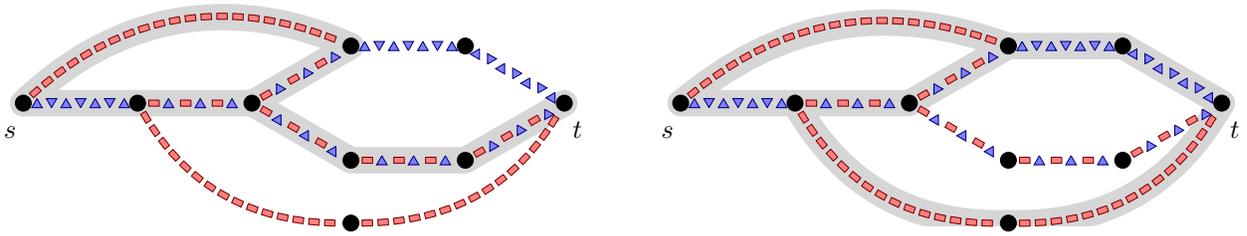
\begin{figure}[!ht]
\centering
\begin{tikzpicture}[scale=.76]
 \tikzstyle{vertex}=[draw,circle,fill,minimum size=6,inner sep=0]

 \tikzset{edge_blue/.style={postaction={decorate, decoration={markings, mark=between positions 2pt and 1-2pt step 11pt with {\draw[blue!50!black,thin,fill=blue!50] ($(90:0.08) - (0,0.02)$) -- ($(210:0.08) - (0,0.02)$) -- ($(330:0.08) - (0,0.02)$) -- ($(90:0.08) - (0,0.02)$);}}}, postaction={decorate, decoration={markings, mark=between positions 7.5pt and 1-2pt step 11pt with {\draw[blue!50!black,thin,fill=blue!50] ($(270:0.08) + (0,0.02)$) -- ($(30:0.08) + (0,0.02)$) -- ($(150:0.08) + (0,0.02)$) -- ($(270:0.08) + (0,0.02)$);}}}}}

 \tikzset{edge_red/.style={postaction={decorate, decoration={markings, mark=between positions 3pt and 1-3pt step 6pt with {\draw[red!50!black,thin,fill=red!50] (30:0.08) -- (150:0.08) -- (210:0.08) -- (330:0.08) -- (30:0.08);}}}}}

 \tikzset{edge_red_and_blue/.style={postaction={decorate, decoration={markings, mark=between positions 3pt and 1-3pt step 12pt with {\draw[red!50!black,thin,fill=red!50] (30:0.08) -- (150:0.08) -- (210:0.08) -- (330:0.08) -- (30:0.08);}}}, postaction={decorate, decoration={markings, mark=between positions 8.5pt and 1-3pt step 12pt with {\draw[blue!50!black,thin,fill=blue!50] ($(90:0.08) - (0,0.02)$) -- ($(210:0.08) - (0,0.02)$) -- ($(330:0.08) - (0,0.02)$) -- ($(90:0.08) - (0,0.02)$);}}}}}

 \node[vertex] (s) at (0,0) [label={[xshift=-5pt, yshift=-20pt]:$s$}] {};
 \node[vertex] (t) at (6 + 2*1.73205080757,0) [label={[xshift=5pt, yshift=-20pt]:$t$}] {};

 \node[vertex] (b1) at (2,0) {};
 \node[vertex] (b2) at (4,0) {};
 \node[vertex] (b3) at (4 + 1.73205080757,1) {};
 \node[vertex] (b4) at (6 + 1.73205080757,1) {};

 \node[vertex] (c1) at (4 + 1.73205080757,-1) {};
 \node[vertex] (c2) at (6 + 1.73205080757,-1) {};

 \node[vertex] (d1) at (4 + 1.73205080757,-1-1.1) {};

 \begin{pgfonlayer}{background}
 \begin{scope}[opacity=.8,transparency group]
  \highlight{10pt}{black!20}{(s.center) to [bend left=0] (b1.center) to [bend left=0] (b2.center) to [bend left=0] (c1.center) to [bend left=0] (c2.center) to [bend left=0] (t.center)}
  \highlight{10pt}{black!20}{(s.center) to [bend left=34] (b3.center) to [bend left=0] (b2.center)}
 \end{scope}
 \end{pgfonlayer}

 \path[edge_blue] (s) -- (b1);
 \path[edge_red_and_blue] (b1) -- (b2);
 \path[edge_red_and_blue] (b2) -- (b3);
 \path[edge_blue] (b3) -- (b4);
 \path[edge_blue] (b4) -- (t);
 \path[edge_red] (s) to [bend left=33] (b3);

 \path[edge_red_and_blue] (b2) -- (c1);
 \path[edge_red_and_blue] (c1) -- (c2);
 \path[edge_red_and_blue] (c2) -- (t);

 \path[edge_red] (b1) to [out=-90*2/3,in=180] (d1);
 \path[edge_red] (d1) to [out=0,in=-90+-90*1/3] (t);

 \begin{scope}[shift={(11.5,0)}]
 \node[vertex] (s) at (0,0) [label={[xshift=-5pt, yshift=-20pt]:$s$}] {};
 \node[vertex] (t) at (6 + 2*1.73205080757,0) [label={[xshift=5pt, yshift=-20pt]:$t$}] {};

 \node[vertex] (b1) at (2,0) {};
 \node[vertex] (b2) at (4,0) {};
 \node[vertex] (b3) at (4 + 1.73205080757,1) {};
 \node[vertex] (b4) at (6 + 1.73205080757,1) {};

 \node[vertex] (c1) at (4 + 1.73205080757,-1) {};
 \node[vertex] (c2) at (6 + 1.73205080757,-1) {};

 \node[vertex] (d1) at (4 + 1.73205080757,-1-1.1) {};

 \begin{pgfonlayer}{background}
 \begin{scope}[opacity=.8,transparency group]
  \highlight{10pt}{black!20}{(s.center) to [bend left=0] (b1.center) to [bend left=0] (b2.center) to [bend left=0] (b3.center) to [bend left=0] (b4.center) to [bend left=0] (t.center)}
  \highlight{10pt}{black!20}{(s.center) to [bend left=30] (b3.center)}
  \highlight{10pt}{black!20}{(b1.center) to [bend right=31] (d1.center) to [bend right=31] (t.center)}
 \end{scope}
 \end{pgfonlayer}

 \path[edge_blue] (s) -- (b1);
 \path[edge_red_and_blue] (b1) -- (b2);
 \path[edge_red_and_blue] (b2) -- (b3);
 \path[edge_blue] (b3) -- (b4);
 \path[edge_blue] (b4) -- (t);
 \path[edge_red] (s) to [bend left=30] (b3);

 \path[edge_red_and_blue] (b2) -- (c1);
 \path[edge_red_and_blue] (c1) -- (c2);
 \path[edge_red_and_blue] (c2) -- (t);

 \path[edge_red] (b1) to [bend right=30] (d1);
 \path[edge_red] (d1) to [bend right=30] (t);
 \end{scope}
\end{tikzpicture}
\caption{Examples for \SupUndSPMinCost and to \ExUndSPMinCost.
  The arcs marked with (red) rectangles and with (blue) triangles correspond to the arcs in $A_1\setminus A_2$ and $A_2\setminus A_1$, respectively.
  The arcs marked with alternating (red) rectangles and (blue) triangles are in $A_1\cap A_2$.
  Each arc has cost $1$.
  The highlighted arcs on the left form an optimal solution to \SupUndSPMinCost, while those on the right form an optimal solution to \ExUndSPMinCost.
  Note that \SupUndSPMinCost has no optimal solution consisting of a shortest red and a shortest blue path.}
\label{fig:example}
\end{figure}

Note that for $k=1$, each of the problems defined above reduces to the classical shortest path problem.
Furthermore, for $k>1$, the natural generalization of the above problems in which each color class has its own terminal pair to be connected can be easily reduced to the corresponding case with a single pair of terminals as follows.
Formally, consider the modifications of the above problems in which multiple sources $s_1, \ldots, s_k$ and multiple targets $t_1, \ldots, t_k$ are given and the task is to find (directed or undirected) $s_i$-$t_i$ paths of color $i$ satisfying the corresponding conditions.
These problems can be reduced to the original ones by adding two new vertices $s$ and $t$, and the arcs or edges $ss_i$ and $t_i t$ with color $i$ to the graph for all $i \in \{ 1, \ldots, k \}$.

\subsection{Previous work}

The above reductions also show that the ``superset'' variants generalize the so-called \textsc{SteinerForest} problem~\cite{Agrawal1995Steiner}, in which we look for a minimum-cost subset of arcs or edges connecting multiple distinct terminal pairs.
Here, the color classes are used to emulate different source--sink pairs as described above, but they do not restrict which arcs or edges can be used to connect distinct terminal pairs.
However, as we establish in this paper, the introduction of restricted link usage makes the problems significantly harder.
For example, while undirected \textsc{SteinerForest} allows for a constant-factor approximation~\cite{Goemans1995General} and both directed and undirected \textsc{SteinerForest} can be solved in polynomial time when the number of terminals is constant~\cite{Feldman2006Directed}, neither {\SupDirSPMinCost} nor {\SupUndSPMinCost} allow for constant-factor approximations in general, and both problems are APX-complete even when only two color classes are present.

Similar problems have been studied before, in which we want to find such a minimum-cost subset of some ground set (e.g., the edge set of a graph) whose certain restrictions satisfy some conditions.
Such a problem concerning $b$-matchings was proved to be APX-complete in general, and polynomial-time solvable in special cases~\cite{madarasi2023sap}.

Now we propose a natural question of this type concerning the independent sets of matroids.
Let $S_1,\ldots,S_k$ be finite, nonempty, not necessarily disjoint sets, and let us be given a matroid $M_i$ on the ground set $S_i$ for all $i\in\{1,\ldots,k\}$.
Our goal is to find a maximum-size subset $F$ of $S_1 \cup \dots \cup S_k$ such that $F\cap S_i$ is independent in $M_i$ for each $i\in\{1,\ldots,k\}$.
This problem is NP-hard for $k \ge 3$ because it includes the matroid intersection problem for 3 matroids.
However, the problem is polynomial-time solvable provided that $k=2$ and all matroids are representable over the same field by a straightforward reduction to the matroid matching problem~\cite{Jenkyns1974Matchoid, Lovasz1990Matroid}.

The framework, in which instead of restricting on a single color only, we restrict to all but one colors, is also called color-avoidance~\cite{Krause2016Hidden}.
The problem of finding minimum-cost color-avoiding connected spanning subgraphs and its natural matroidal generalization were shown to be NP-hard but to admit a polynomial-time 2-approximation algorithm~\cite{Pinter2024Color}.

\subsection{Our contributions}

In this paper, we study the variants of the simultaneous path problem introduced above; for a summary of our results, see Table~\ref{table:summary}.

\begin{table}[!ht]
\begin{center}
 \setlength{\tabcolsep}{5pt}
 \renewcommand{\arraystretch}{1.2}
 \captionsetup{type=table}
 \begin{tabular}{!{\vrule width 1.25pt} c !{\vrule width 1.25pt} c !{\vrule width 1.25pt} c !{\vrule width 1.25pt}}
   \specialrule{1.25pt}{0pt}{0pt}
   & \scshape{Exact-Simultaneous-} & \scshape{Superset-Simultaneous-} \\ \specialrule{1.25pt}{0pt}{0pt}
   \multirow{8}{*}{\parbox{60pt}{\centering {\scshape Dipaths} \\ feasibility}} & \multirow{5}{*}{\parbox{160pt}{\centering NP-complete \\ even for $k=2$ and $\ell = 1$ \\ (Thm.~\ref{thm:exact_directed_NPC}) \\ and even in DAGs \\ (Thm.~\ref{thm:exact_DAG_NPC})}} & \multirow{8}{*}{\parbox{160pt}{\centering in P \\ (trivial)}} \\
   & & \\
   & & \\
   & & \\
   & & \\ \cline{2-2}
   & \multirow{3}{*}{\parbox{150pt}{\centering FPT parameterized by $\ell$\\ in planar digraphs\\ (Cor.~\ref{cor:ExDirSPEx_planar})}} & \\
   & & \\
   & & \\ \specialrule{1.25pt}{0pt}{0pt}
   \multirow{12}{*}{\parbox{60pt}{\centering \scshape{Dipaths}}} & \multirow{5}{*}{\parbox{160pt}{\centering no constant approx.\\ even for $k=2$ and $\ell = 1$ \\
   (Thm.~\ref{thm:exact_directed_inapprox})}} & \multirow{3}{*}{\parbox{160pt}{\centering $\big( (1-\varepsilon) \ln k \big)$-inapprox.~$\forall \varepsilon > 0$ \\ even for $c \equiv 1$ in DAGs \\ (Thm.~\ref{thm:superset_DAG_APXC})}} \\
   & & \\
   & & \\ \cline{3-3}
   & & \multirow{3}{*}{\parbox{160pt}{\centering APX-complete \\ even for $k=2$ and $c \equiv 1$ \\ (Thm.~\ref{thm:superset_directed_APXC})}} \\
   & & \\ \cline{2-2}
   & \multirow{5}{*}{\parbox{160pt}{\centering in P for fixed $k$ in DAGs \\ (Thm.~\ref{thm:exact_dac_constantk})}} & \\ \cline{3-3}
   & & \multirow{2}{*}{\parbox{160pt}{\centering FPT parameterized by $\ell$ \\ (Thm.~\ref{thm:intersectionFPT_SupDir})}} \\
   & & \\ \cline{3-3}
   & & \multirow{2}{*}{\parbox{160pt}{\centering in P for fixed $k$ in DAGs \\ (Thm.~\ref{thm:sup_dac_constantk})}} \\
   & & \\ \cline{2-3}
   & \multicolumn{2}{c !{\vrule width 1.25pt}}{\multirow{2}{*}{\parbox{320pt}{\centering in P if the color classes are laminar \\ (Thm.~\ref{thm:laminarfamily})}}} \\
   & \multicolumn{2}{c !{\vrule width 1.25pt}}{} \\
   \specialrule{1.25pt}{0pt}{0pt}
   \multirow{5}{*}{\parbox{60pt}{\centering {\scshape Paths} \\ feasibility}} & \multirow{3}{*}{\parbox{160pt}{\centering NP-complete \\ even for $k=2$ \\ (Thm.~\ref{thm:exact_undirected_NPC})}} & \multirow{5}{*}{\parbox{160pt}{\centering in P \\ (trivial)}} \\
   & & \\
   & & \\ \cline{2-2}
   & \multirow{2}{*}{\parbox{160pt}{\centering FPT parameterized by $\ell$ \\ (Thm.~\ref{thm:ExUndSPEx_fpt})}} & \\
   & & \\ \specialrule{1.25pt}{0pt}{0pt}
   \multirow{10}{*}{\parbox{60pt}{\centering \scshape{Paths}}} & \multirow{8}{*}{\parbox{160pt}{\centering no constant approx.\\ even for $k=2$ \\
   (Thm.~\ref{thm:exact_undirected_inapprox})}} & \multirow{3}{*}{\parbox{160pt}{\centering $\big( (1-\varepsilon) \ln k \big)$-inapprox.~$\forall \varepsilon > 0$ \\ even for $c \equiv 1$ \\ (Thm.~\ref{thm:superset_undirected_inapprox})}} \\
   & & \\
   & & \\ \cline{3-3}
   & & \multirow{3}{*}{\parbox{160pt}{\centering APX-complete \\ even for $k=2$ and $c \equiv 1$ \\ (Thm.~\ref{thm:super_undirected_APXC})}} \\
   & & \\
   & & \\ \cline{3-3}
   & & \multirow{2}{*}{\parbox{160pt}{\centering FPT parameterized by $\ell$ \\ (Thm.~\ref{thm:intersectionFPT_SupDir})}} \\
   & & \\ \cline{2-3}
   & \multicolumn{2}{c !{\vrule width 1.25pt}}{\multirow{2}{*}{\parbox{320pt}{\centering in P if the color classes are laminar \\ (Thm.~\ref{thm:laminarfamily})}}} \\
   & \multicolumn{2}{c !{\vrule width 1.25pt}}{} \\
   \specialrule{1.25pt}{0pt}{0pt}
 \end{tabular}
 \captionof{table}{Summary of our results, where $\ell$ denotes the number of multi-colored arcs or edges.}
 \label{table:summary}
\end{center}
\end{table}

We derive strong hardness results even for restricted cases, in particular, we prove that all variants are intractable.
Specifically, the decision version of the ``exact'' variants, \ExDirSPMinCost and \ExUndSPMinCost, are both NP-complete even when the number $k$ of colors is $2$ and the cost function is uniform; see Theorems~\ref{thm:exact_directed_NPC} and~\ref{thm:exact_undirected_NPC}.
The optimization versions of these problems cannot be approximated within a constant factor unless $\text{P} = \text{NP}$; see Theorems~\ref{thm:exact_directed_inapprox}~and~\ref{thm:exact_undirected_inapprox}.
We also show that the \ExDirSPMinCost problem remains NP-complete even for directed acyclic graphs; see Theorem~\ref{thm:exact_DAG_NPC}.

In contrast, it is easy to decide the feasibility of the ``superset'' variants: they are feasible if and only if the whole arc or edge set forms a feasible solution, that is, the subgraph $D_i=(V,A_i)$ or $G_i=(V,E_i)$ contains a directed or undirected $s$-$t$ path for all $i \in \{ 1, \ldots, k \}$.
However, we show that \SupDirSPMinCost and \SupUndSPMinCost are both APX-complete even when the number $k$ of colors is 2; see Theorems~\ref{thm:superset_directed_APXC} and~\ref{thm:super_undirected_APXC}.
In addition, we show that, unless $\text{P} = \text{NP}$, the \SupDirSPMinCost problem in directed acyclic graphs is inapproximable within a factor of $(1-\varepsilon) \ln k$ for any constant $\varepsilon > 0$; see Theorem~\ref{thm:superset_DAG_APXC}.

On the positive side, one can solve \ExDirSPMinCost and \SupDirSPMinCost for directed acyclic graphs when the number of color classes is constant; see Theorems~\ref{thm:exact_dac_constantk} and~\ref{thm:sup_dac_constantk}.
Moreover, we present an FPT algorithm for \SupDirSPMinCost and \SupUndSPMinCost parameterized by the number of multi-colored arcs or edges; see Theorem~\ref{thm:intersectionFPT_SupDir}.
We also propose an FPT algorithm with the same parameter to decide the feasibility of \SupUndSPMinCost; see Theorem~\ref{thm:ExUndSPEx_fpt}.
To decide the feasibility of \SupDirSPMinCost in restricted cases, we give an FPT algorithm with the same parameter.
Finally, we observe that all variants can be solved in polynomial time when the color classes form a laminar family; see Theorem~\ref{thm:laminarfamily}.

\section{Hardness results}

In this section, we derive our hardness results.

\subsection{Directed variants}\label{sec:reductionFromSat}

We now prove that the following decision version of the {\ExDirSPMinCost} problem in which one is to decide whether a feasible solution exists is NP-complete.

\smallskip

\noindent {\bf \ExDirSPEx} \\
\textit{Input:} a positive integer $k$, a digraph~$D=(V,A)$ with $A = A_1 \cup \cdots \cup A_k$, and two distinct vertices $s, t \in V$.\\
\textit{Problem:} decide whether there exists $A' \subseteq A$ such that $A' \cap A_i$ is an $s$-$t$ dipath for all $i \in \{ 1, \ldots, k \}$.

\medskip

\begin{thm}\label{thm:exact_directed_NPC}
The {\ExDirSPEx} problem is NP-complete even when $k=2$ and $|A_1\cap A_2|=1$.
\end{thm}
\begin{proof}
The problem is clearly in NP.
We give a reduction from the NP-complete \textsc{TwoDisjointDipaths} problem~\cite{Fortune1980Directed} to {\ExDirSPEx} for $k=2$.
The input of the \textsc{TwoDisjointDipaths} problem is an arbitrary digraph $D = (V,A)$ along with four distinct given vertices $s_1, t_1, s_2, t_2 \in V$, and our goal is to find an $s_i$-$t_i$ dipath for each $i\in\{1,2\}$ that are vertex-disjoint.

Without loss of generality, we assume that $t_1s_2 \in A$ and $s_1t_1, s_2t_2 \notin A$.
We construct an instance of the {\ExDirSPEx} problem for $k=2$ as follows.
Define $A_1 = A$ and $A_2 = \{ s_1t_1, t_1s_2, s_2t_2 \}$.
Consider the digraph $\widetilde{D} = (V, A_1 \cup A_2)$ with two specified vertices $s=s_1$ and $t=t_2$, which clearly has polynomial size in that of $D$.

We show below that $D$ contains an $s_i$-$t_i$ dipath for each $i\in\{1,2\}$ that are vertex-disjoint if and only if there exists $A' \subseteq A_1 \cup A_2$ such that both $A' \cap A_1$ and $A' \cap A_2$ are $s$-$t$ dipaths.

First, assume that $D$ contains an $s_i$-$t_i$ dipath for each $i\in\{1,2\}$ that are vertex-disjoint.
Then let $A'$ be the set consisting of the arcs of these dipaths and the arcs of $A_2$.
By the definition of $A'$, both $A' \cap A_1$ and $A' \cap A_2$ are $s$-$t$ dipaths in $\widetilde{D}$.
Conversely, assume that there exists $A' \subseteq A_1 \cup A_2$ such that both $A' \cap A_1$ and $A' \cap A_2$ are $s$-$t$ dipaths in $\widetilde{D}$.
Since $A_2$ is an $s$-$t$ dipath, $A_2 \subseteq A'$ must hold, and thus the arc $s_1t_2$ is contained in $A'$.
Hence, the dipath $A' \cap A_1$ must use the arc $s_1t_2$, which means that $(A'\cap A_1)\setminus \{s_1t_2\}$ is the disjoint union of an $s_1$-$t_1$ and an $s_2$-$t_2$ dipath.
\end{proof}

This immediately implies that the optimization counterpart cannot be approximated unless $\text{P} = \text{NP}$.

\begin{thm}\label{thm:exact_directed_inapprox}
There exists no polynomial-time $\alpha$-approximation algorithm for the \ExDirSPMinCost problem for any $\alpha \geq 1$ unless $\text{P} = \text{NP}$, even if $k=2$ and $|A_1\cap A_2|=1$.
\end{thm}
\begin{proof}
Consider an instance of the {\ExDirSPEx} problem consisting of a digraph $D = (V,A)$, where the arc set $A$ is the union of two subsets $A_1$ and $A_2$ with $|A_1 \cap A_2| = 1$, and where $s,t \in V$ are two distinguished vertices.
Add two parallel arcs $a_1$ and $a_2$ from $s$ to $t$, and let $A'_1 = A_1 \cup \{ a_1 \}$ and $A'_2 = A_2 \cup \{ a_2 \}$.
Let $D'$ denote the so-obtained digraph, which has clearly polynomial size.
Let $c(a)$ be $1$ if $a\in \{a_1,a_2\}$, and $0$ otherwise.

Note that there exists a feasible solution to \ExDirSPMinCost in $D'$.
Given an $\alpha$-approximation algorithm for \ExDirSPMinCost, where $\alpha \geq 1$, it returns a feasible solution of $c$-cost $0$ in $D'$ if and only if \ExDirSPEx is feasible for $D$.
By Theorem~\ref{thm:exact_directed_NPC}, this implies $\text{P} = \text{NP}$.
\end{proof}

In contrast to Theorem~\ref{thm:exact_directed_NPC} above, we can find a feasible solution to \SupDirSPMinCost in polynomial time by finding an $s$-$t$ dipath in each color class.
Moreover, we can get a polynomial-time $k$-approximation algorithm for this problem as follows.
Specifically, let $P_i$ be a shortest $s$-$t$ dipath in each color class $i \in \{ 1, \ldots, k \}$.
Then the union of all these dipaths forms a $k$-approximate solution.
Indeed, since the length of such a shortest dipath is a lower bound on the optimum value of \SupDirSPMinCost, the sum of the length of these shortest dipaths is at most $k$ times the optimum value of \SupDirSPMinCost.
Clearly, a two-vertex digraph, with $k+1$ parallel arcs $a_1, \ldots, a_{k+1}$ of cost 1, where $A_i = \{ a_i, a_{k+1} \}$ for each $i \in \{ 1, \ldots, k \}$ gives a tight example.
However, the following theorem shows that there exists no polynomial-time approximation scheme unless $\text{P} = \text{NP}$.

\begin{thm}\label{thm:superset_directed_APXC}
The \SupDirSPMinCost problem is APX-complete even for $c \equiv 1$ and $k=2$.
\end{thm}
\begin{proof}
Since we proposed a polynomial-time $k$-approximation algorithm to \SupDirSPMinCost above, the problem is in APX.
To show that it is APX-hard, we reduce the problem MAX-2SAT3 to it.
The input of MAX-2SAT3 is a Boolean formula in conjunctive normal form where each clause has size (at most) 2 and each variable occurs at most 3 times, and our goal is to find an assignment which maximizes the number of satisfied clauses.
It was shown in~\cite{Berman1999Some} that MAX-2SAT3 is NP-hard to approximate within a factor of $2011/2012 + \varepsilon$ for any $\varepsilon > 0$.
Now we show that if there were a polynomial-time ($34205/34204-\varepsilon$)-approximation algorithm for the \SupDirSPMinCost problem for some $\varepsilon > 0$, then we would obtain a polynomial-time ($2011/2012 + \varepsilon'$)-approximation algorithm for the MAX-2SAT3 problem for some $\varepsilon'>0$, implying $\text{P} = \text{NP}$.

Let $\varphi = C_1 \wedge \cdots \wedge C_m$ be an instance of the MAX-2SAT3 problem, where $C_1,\ldots,C_m$ are the clauses of $\varphi$, and let $x_1,\ldots,x_n$ denote its variables.
Without loss of generality, we assume that every variable appears both negated and non-negated, and also that each variable appears at most once per clause.

We define a digraph $D=(V,A)$ as follows.
For each $i \in \{ 1, \ldots, n \}$, take the distinct vertices $u_{i1}, \ldots, u_{i4}, \allowbreak v_{i1}, \ldots, v_{i4}$, and $w_i$, and then add the arcs of the dipaths  $P_i^+ = (w_i, v_{i1}, \ldots,\allowbreak v_{i4}, w_{i+1})$ and $P_i^- = (w_i, u_{i1}, \ldots, \allowbreak u_{i4}, w_{i+1})$ associated to the literals $x_i$ and $\neg x_i$, respectively.
For each $j \in \{ 1, \ldots, m \}$, add a new vertex $c_j$ associated to the clause $C_j$ and also add a new vertex $c_{m+1}$ to the digraph.
For each $i \in \{ 1, \ldots, n \}$, consider the occurrences of the variable $x_i$ in $\varphi$.
Let $j_1, j_2 \in \{ 1, \ldots, m \}$ be the smallest indices for which $x_i$ is non-negated in the clause $C_{j_1}$ and is negated in $C_{j_2}$, respectively.
Then add the arcs $c_{j_1} v_{i1}$, $v_{i2} c_{j_1+1}$ and $c_{j_2} u_{i1}$, $u_{i2} c_{j_2+1}$ to the digraph.
If $x_i$ occurs 3 times, then let $j_3 \in \{ 1, \ldots, m \}$ denote the index of the clause other than $C_{j_1}$ and $C_{j_2}$ in which $x_i$ also appears.
If $x_i$ appears non-negated in $C_{j_3}$, then add the arcs $c_{j_3} v_{i3}$, $v_{i4} c_{j_3+1}$ to the digraph; if $x_i$ appears negated in $C_{j_3}$, then add the arcs $c_{j_3} u_{i3}$, $u_{i4} c_{j_3+1}$ to the digraph.
Add two new distinct vertices $s$ and $t$, and add the arcs $sw_1, sc_1$ and $w_{n+1}t, c_{m+1}t$ to the digraph.

Let $A_1$ consist of the arcs $sw_1, w_{n+1}t$, and the arcs of the dipaths $P_i^+$ and $P_i^-$ for each $i \in \{ 1, \ldots, n \}$.
Let $A_2$ consist of all arcs incident to some vertices of $\{ c_1, \ldots, c_{m+1} \}$ and the arcs $v_{i1}v_{i2}, v_{i3}v_{i4}, u_{i1}u_{i2}, u_{i3}u_{i4}$ for each $i \in \{ 1, \ldots, n \}$.
For an example, see Figure~\ref{fig:construction_SupDirSP}.
Clearly, the new instance has polynomial size in the number of clauses and variables of the Boolean formula.

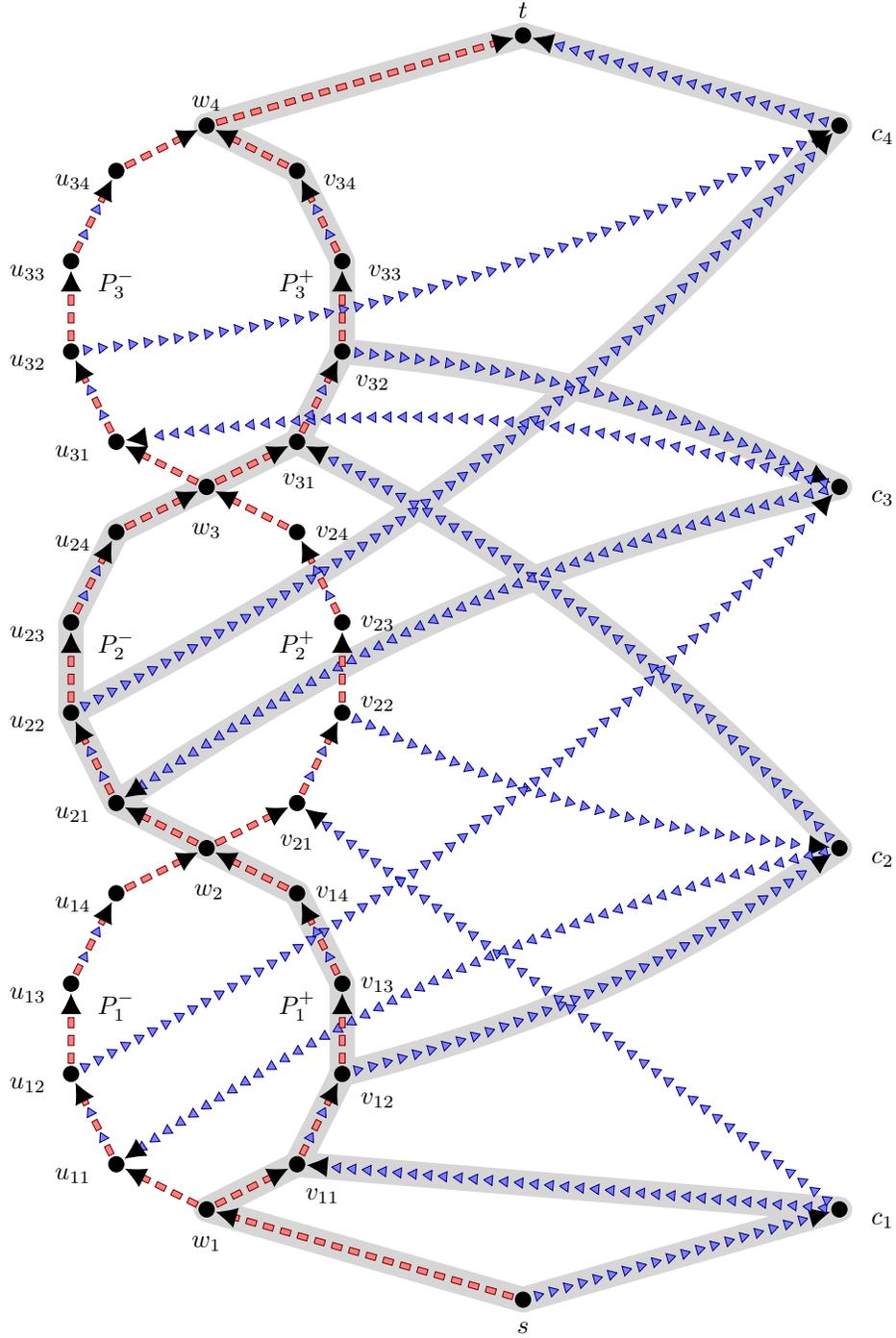
\begin{figure}[!ht]
\centering
\begin{tikzpicture}[scale=0.625]
 \tikzstyle{vertex}=[draw,circle,fill,minimum size=6,inner sep=0]

 \tikzset{edge_blue/.style={postaction={decorate, decoration={markings, mark=between positions 3pt and 1-10pt step 6pt with {\draw[blue!50!black,thin,fill=blue!50] (0:0.08) -- (120:0.08) -- (240:0.08) -- (0:0.08);}}}}}

 \tikzset{edge_red/.style={postaction={decorate, decoration={markings, mark=between positions 3pt and 1-8pt step 7pt with {\draw[red!50!black,thin,fill=red!50] (30:0.08) -- (150:0.08) -- (210:0.08) -- (330:0.08) -- (30:0.08);}}}}}

 \tikzset{edge_red_and_blue/.style={postaction={decorate, decoration={markings, mark=between positions 3pt and 1-8pt step 11pt with {\draw[red!50!black,thin,fill=red!50] (30:0.08) -- (150:0.08) -- (210:0.08) -- (330:0.08) -- (30:0.08);}}}, postaction={decorate, decoration={markings, mark=between positions 8.5pt and 1-8pt step 11pt with {\draw[blue!50!black,thin,fill=blue!50] (0:0.08) -- (120:0.08) -- (240:0.08) -- (0:0.08);}}}}}

 \node at (-2,4.5) {$P^-_1$};
 \node at (-2,12.5) {$P^-_2$};
 \node at (-2,20.5) {$P^-_3$};
 \node at (2,4.5) {$P^+_1$};
 \node at (2,12.5) {$P^+_2$};
 \node at (2,20.5) {$P^+_3$};

 \node[vertex] (s) at (7,-2) [label={[xshift=0pt, yshift=-20pt]:$s$}] {};
 \node[vertex] (t) at (7,26) [label={[xshift=0pt, yshift=0pt]:$t$}] {};

 \node[vertex] (w1) at (0,0) [label={[xshift=0pt, yshift=-23pt]:$w_1$}] {};
 \node[vertex] (u11) at (-2,1) [label={[xshift=-17pt, yshift=-14pt]:$u_{11}$}] {};
 \node[vertex] (u12) at (-3,3) [label={[xshift=-17pt, yshift=-14pt]:$u_{12}$}] {};
 \node[vertex] (u13) at (-3,5) [label={[xshift=-17pt, yshift=-14pt]:$u_{13}$}] {};
 \node[vertex] (u14) at (-2,7) [label={[xshift=-17pt, yshift=-14pt]:$u_{14}$}] {};
 \node[vertex] (v11) at (2,1) [label={[xshift=10pt, yshift=-22pt]:$v_{11}$}] {};
 \node[vertex] (v12) at (3,3) [label={[xshift=14pt, yshift=-20pt]:$v_{12}$}] {};
 \node[vertex] (v13) at (3,5) [label={[xshift=14pt, yshift=-10pt]:$v_{13}$}] {};
 \node[vertex] (v14) at (2,7) [label={[xshift=14pt, yshift=-10pt]:$v_{14}$}] {};
 \node[vertex] (w2) at (0,8) [label={[xshift=1pt, yshift=-27pt]:$w_2$}] {};
 \node[vertex] (u21) at (-2,9) [label={[xshift=-17pt, yshift=-14pt]:$u_{21}$}] {};
 \node[vertex] (u22) at (-3,11) [label={[xshift=-17pt, yshift=-14pt]:$u_{22}$}] {};
 \node[vertex] (u23) at (-3,13) [label={[xshift=-17pt, yshift=-14pt]:$u_{23}$}] {};
 \node[vertex] (u24) at (-2,15) [label={[xshift=-17pt, yshift=-14pt]:$u_{24}$}] {};
 \node[vertex] (v21) at (2,9) [label={[xshift=0pt, yshift=-25pt]:$v_{21}$}] {};
 \node[vertex] (v22) at (3,11) [label={[xshift=14pt, yshift=-6pt]:$v_{22}$}] {};
 \node[vertex] (v23) at (3,13) [label={[xshift=14pt, yshift=-10pt]:$v_{23}$}] {};
 \node[vertex] (v24) at (2,15) [label={[xshift=14pt, yshift=-10pt]:$v_{24}$}] {};
 \node[vertex] (w3) at (0,16) [label={[xshift=0pt, yshift=-27pt]:$w_3$}] {};
 \node[vertex] (u31) at (-2,17) [label={[xshift=-17pt, yshift=-14pt]:$u_{31}$}] {};
 \node[vertex] (u32) at (-3,19) [label={[xshift=-17pt, yshift=-14pt]:$u_{32}$}] {};
 \node[vertex] (u33) at (-3,21) [label={[xshift=-17pt, yshift=-14pt]:$u_{33}$}] {};
 \node[vertex] (u34) at (-2,23) [label={[xshift=-17pt, yshift=-14pt]:$u_{34}$}] {};
 \node[vertex] (v31) at (2,17) [label={[xshift=1pt, yshift=-25pt]:$v_{31}$}] {};
 \node[vertex] (v32) at (3,19) [label={[xshift=12pt, yshift=-22pt]:$v_{32}$}] {};
 \node[vertex] (v33) at (3,21) [label={[xshift=17pt, yshift=-14pt]:$v_{33}$}] {};
 \node[vertex] (v34) at (2,23) [label={[xshift=17pt, yshift=-14pt]:$v_{34}$}] {};
 \node[vertex] (w4) at (0,24) [label={[xshift=0pt, yshift=-1pt]:$w_4$}] {};

 \node[vertex] (c1) at (14,0) [label={[xshift=17pt, yshift=-14pt]:$c_1$}] {};
 \node[vertex] (c2) at (14,8) [label={[xshift=17pt, yshift=-14pt]:$c_2$}] {};
 \node[vertex] (c3) at (14,16) [label={[xshift=17pt, yshift=-14pt]:$c_3$}] {};
 \node[vertex] (c4) at (14,24) [label={[xshift=17pt, yshift=-14pt]:$c_4$}] {};

 \begin{pgfonlayer}{background}
      \begin{scope}[opacity=.8,transparency group]
        \highlight{10pt}{black!20}{(s.center) to [bend left =0] (w1.center) to [bend left =0] (v11.center) to [bend left =0] (v12.center) to [bend left =0] (v13.center) to [bend left =0] (v14.center) to [bend left =0] (w2.center) to [bend left =0] (u21.center) to [bend left =0] (u22.center) to [bend left =0] (u23.center) to [bend left =0] (u24.center) to [bend left =0] (w3.center) to [bend left =0] (v31.center) to [bend left =0] (v32.center) to [bend left =0] (v33.center) to [bend left =0] (v34.center) to [bend left =0] (w4.center) to [bend left =0] (t.center)}
        \highlight{10pt}{black!20}{(s.center) to [bend left =0] (c1.center) to [bend left =0] (v11.center) to [bend left =0] (v12.center) to [bend right =10] (c2.center) to [bend right =10] (v31.center) to [bend left =0] (v32.center) to [bend left =10] (c3.center) to [bend right=10] (u21.center) to [bend left =0] (u22.center) to [bend right=10] (c4.center) to [bend left =0] (t.center)}
      \end{scope}
 \end{pgfonlayer}

 \path[edge_red] (s) -- (w1);
 \draw[->,>={LaTeX[black,length=9pt, width=9pt]},draw opacity=0] (s) -- (w1);
 \path[edge_blue] (s) -- (c1);
 \draw[->,>={LaTeX[black,length=9pt, width=9pt]},draw opacity=0] (s) -- (c1);
 \path[edge_red] (w4) -- (t);
 \draw[->,>={LaTeX[black,length=9pt, width=9pt]},draw opacity=0] (w4) -- (t);
 \path[edge_blue] (c4) -- (t);
 \draw[->,>={LaTeX[black,length=9pt, width=9pt]},draw opacity=0] (c4) -- (t);

 \path[edge_blue] (c1) -- (v11);
 \draw[->,>={LaTeX[black,length=9pt, width=9pt]},draw opacity=0] (c1) -- (v11);
 \path[edge_blue] (v12) to [bend right=10] (c2);
 \draw[->,>={LaTeX[black,length=9pt, width=9pt]},draw opacity=0] (v12) to [bend right=10] (c2);
 \path[edge_blue] (c1) -- (v21);
 \draw[->,>={LaTeX[black,length=9pt, width=9pt]},draw opacity=0] (c1) -- (v21);
 \path[edge_blue] (v22) to [bend right=10] (c2);
 \draw[->,>={LaTeX[black,length=9pt, width=9pt]},draw opacity=0] (v22) to [bend right=10] (c2);

 \path[edge_blue] (c2) to [bend right=10] (u11);
 \draw[->,>={LaTeX[black,length=9pt, width=9pt]},draw opacity=0] (c2) to [bend right=10] (u11);
 \path[edge_blue] (u12) to [bend right=10] (c3);
 \draw[->,>={LaTeX[black,length=9pt, width=9pt]},draw opacity=0] (u12) to [bend right=10] (c3);
 \path[edge_blue] (c2) to [bend right=10] (v31);
 \draw[->,>={LaTeX[black,length=9pt, width=9pt]},draw opacity=0] (c2) to [bend right=10] (v31);
 \path[edge_blue] (v32) to [bend left=10] (c3);
 \draw[->,>={LaTeX[black,length=9pt, width=9pt]},draw opacity=0] (v32) to [bend left=10] (c3);

 \path[edge_blue] (c3) to [bend right=10] (u21);
 \draw[->,>={LaTeX[black,length=9pt, width=9pt]},draw opacity=0] (c3) to [bend right=10] (u21);
 \path[edge_blue] (u22) to [bend right=10] (c4);
 \draw[->,>={LaTeX[black,length=9pt, width=9pt]},draw opacity=0] (u22) to [bend right=10] (c4);
 \path[edge_blue] (c3) to [bend right=12] (u31);
 \draw[->,>={LaTeX[black,length=9pt, width=9pt]},draw opacity=0] (c3) to [bend right=12] (u31);
 \path[edge_blue] (u32) to [bend right=10] (c4);
 \draw[->,>={LaTeX[black,length=9pt, width=9pt]},draw opacity=0] (u32) to [bend right=10] (c4);

 \path[edge_red] (w1) -- (u11);
 \draw[->,>={LaTeX[black,length=9pt, width=9pt]},draw opacity=0] (w1) -- (u11);
 \path[edge_red_and_blue] (u11) -- (u12);
 \draw[->,>={LaTeX[black,length=9pt, width=9pt]},draw opacity=0] (u11) -- (u12);
 \path[edge_red] (u12) -- (u13);
 \draw[->,>={LaTeX[black,length=9pt, width=9pt]},draw opacity=0] (u12) -- (u13);
 \path[edge_red_and_blue] (u13) -- (u14);
 \draw[->,>={LaTeX[black,length=9pt, width=9pt]},draw opacity=0] (u13) -- (u14);
 \path[edge_red] (u14) -- (w2);
 \draw[->,>={LaTeX[black,length=9pt, width=9pt]},draw opacity=0] (u14) -- (w2);
 \path[edge_red] (w1) -- (v11);
 \draw[->,>={LaTeX[black,length=9pt, width=9pt]},draw opacity=0] (w1) -- (v11);
 \path[edge_red_and_blue] (v11) -- (v12);
 \draw[->,>={LaTeX[black,length=9pt, width=9pt]},draw opacity=0] (v11) -- (v12);
 \path[edge_red] (v12) -- (v13);
 \draw[->,>={LaTeX[black,length=9pt, width=9pt]},draw opacity=0] (v12) -- (v13);
 \path[edge_red_and_blue] (v13) -- (v14);
 \draw[->,>={LaTeX[black,length=9pt, width=9pt]},draw opacity=0] (v13) -- (v14);
 \path[edge_red] (v14) -- (w2);
 \draw[->,>={LaTeX[black,length=9pt, width=9pt]},draw opacity=0] (v14) -- (w2);
 \path[edge_red] (w2) -- (u21);
 \draw[->,>={LaTeX[black,length=9pt, width=9pt]},draw opacity=0] (w2) -- (u21);
 \path[edge_red_and_blue] (u21) -- (u22);
 \draw[->,>={LaTeX[black,length=9pt, width=9pt]},draw opacity=0] (u21) -- (u22);
 \path[edge_red] (u22) -- (u23);
 \draw[->,>={LaTeX[black,length=9pt, width=9pt]},draw opacity=0] (u22) -- (u23);
 \path[edge_red_and_blue] (u23) -- (u24);
 \draw[->,>={LaTeX[black,length=9pt, width=9pt]},draw opacity=0] (u23) -- (u24);
 \path[edge_red] (u24) -- (w3);
 \draw[->,>={LaTeX[black,length=9pt, width=9pt]},draw opacity=0] (u24) -- (w3);
 \path[edge_red] (w2) -- (v21);
 \draw[->,>={LaTeX[black,length=9pt, width=9pt]},draw opacity=0] (w2) -- (v21);
 \path[edge_red_and_blue] (v21) -- (v22);
 \draw[->,>={LaTeX[black,length=9pt, width=9pt]},draw opacity=0] (v21) -- (v22);
 \path[edge_red] (v22) -- (v23);
 \draw[->,>={LaTeX[black,length=9pt, width=9pt]},draw opacity=0] (v22) -- (v23);
 \path[edge_red_and_blue] (v23) -- (v24);
 \draw[->,>={LaTeX[black,length=9pt, width=9pt]},draw opacity=0] (v23) -- (v24);
 \path[edge_red] (v24) -- (w3);
 \draw[->,>={LaTeX[black,length=9pt, width=9pt]},draw opacity=0] (v24) -- (w3);
 \path[edge_red] (w3) -- (u31);
 \draw[->,>={LaTeX[black,length=9pt, width=9pt]},draw opacity=0] (w3) -- (u31);
 \path[edge_red_and_blue] (u31) -- (u32);
 \draw[->,>={LaTeX[black,length=9pt, width=9pt]},draw opacity=0] (u31) -- (u32);
 \path[edge_red] (u32) -- (u33);
 \draw[->,>={LaTeX[black,length=9pt, width=9pt]},draw opacity=0] (u32) -- (u33);
 \path[edge_red_and_blue] (u33) -- (u34);
 \draw[->,>={LaTeX[black,length=9pt, width=9pt]},draw opacity=0] (u33) -- (u34);
 \path[edge_red] (u34) -- (w4);
 \draw[->,>={LaTeX[black,length=9pt, width=9pt]},draw opacity=0] (u34) -- (w4);
 \path[edge_red] (w3) -- (v31);
 \draw[->,>={LaTeX[black,length=9pt, width=9pt]},draw opacity=0] (w3) -- (v31);
 \path[edge_red_and_blue] (v31) -- (v32);
 \draw[->,>={LaTeX[black,length=9pt, width=9pt]},draw opacity=0] (v31) -- (v32);
 \path[edge_red] (v32) -- (v33);
 \draw[->,>={LaTeX[black,length=9pt, width=9pt]},draw opacity=0] (v32) -- (v33);
 \path[edge_red_and_blue] (v33) -- (v34);
 \draw[->,>={LaTeX[black,length=9pt, width=9pt]},draw opacity=0] (v33) -- (v34);
 \path[edge_red] (v34) -- (w4);
 \draw[->,>={LaTeX[black,length=9pt, width=9pt]},draw opacity=0] (v34) -- (w4);
\end{tikzpicture}
\caption{The digraph constructed in the proof of Theorem~\ref{thm:superset_directed_APXC} to the Boolean formula $\varphi = (x_1 \lor x_2) \land (\neg x_1 \lor x_3) \land (\neg x_2 \lor \neg x_3)$.
  The arcs marked with (red) rectangles and with (blue) triangles correspond to the arcs in $A_1\setminus A_2$ and $A_2\setminus A_1$, respectively.
  The arcs marked with alternating (red) rectangles and (blue) triangles are in $A_1\cap A_2$.
  The highlighted arcs form the subgraph $A'$ corresponding to the truth assignment $x_1 = x_3 = \texttt{true}$ and $x_2 = \texttt{false}$.}
\label{fig:construction_SupDirSP}
\end{figure}

Let $m_{\mathrm{s}}^*$ denote the maximum number of satisfied clauses under any truth assignment of $\varphi$, and let $A^* \subseteq A$ be an optimal solution of \SupDirSPMinCost in $D$.
Then both $A^* \cap A_1$ and $A^* \cap A_2$ contains an $s$-$t$ dipath and $A^*$ has minimum size.

\begin{claim} \label{cl:SdspToSat}
We have $m - m_{\mathrm{s}}^* \ge |A^*| - (5n + 2m + 4)$.
\end{claim}
\begin{proof}
Consider an optimal truth assignment of $\varphi$.
For each $j \in \{ 1, \ldots, m \}$, let us define an index $i_j \in \{ 1, \ldots, n \}$ as follows.
If $C_j$ is satisfied, then let $i_j$ be the index of the variable corresponding to the first true literal occurring in $C_j$.
Otherwise, let $i_j$ be the index of the variable corresponding to the first literal occurring in $C_j$.
Now we construct a feasible solution $A'$ of \SupDirSPMinCost.
Let $A'$ contain all the arcs incident to $s$ or $t$.
For each $i \in \{ 1, \ldots, n \}$, add the arcs of $P_i^+$ to $A'$ if $x_i$ is true, otherwise add the arcs of $P_i^-$ to $A'$.
For each $j \in \{ 1, \ldots, m \}$, add the unique arc $a_{j1}$ going from the vertex $c_j$ to the vertices of $P_{i_j}^+ \cup P_{i_j}^-$, and add the unique arc $a_{j2}$ going from the vertices of $P_{i_j}^+ \cup P_{i_j}^-$ to the vertex $c_{j+1}$ to $A'$.
Add the arc going from the head of $a_{j1}$ to the tail of $a_{j2}$ to $A'$.
Note that $A'$ contains unique $s$-$t$ dipaths $P_1 \subseteq A_1$ and $P_2 \subseteq A_2$, respectively.
Moreover, $A' = P_1 \cup P_2$.
By construction, for each $j \in \{ 1, \ldots, m \}$, this assignment does not satisfy the clause $C_j$ if and only if the sub-dipath of $P_2$ from $c_j$ to $c_{j+1}$ does not intersect $P_1$.
Then
\begin{multline*}
 m - m^*_{\mathrm{s}} = \big| (P_2 \setminus P_1) \cap (A_1 \cap A_2) \big| = |A'| - \big| A' \setminus (A_1 \cap A_2) \big| - \big| P_1 \cap (A_1 \cap A_2) \big| \\ = |A'| - (3n+2m+4) - 2n = |A'| - (5n+2m+4) \ge |A^*| - (5n+2m+4),
\end{multline*}
which proves the claim.
\end{proof}

\begin{claim}\label{cl:SatToSdsp}
  Let $A'$ be an arbitrary inclusion-wise minimal feasible solution of \SupDirSPMinCost in the constructed digraph $D$.
  Then there exists a truth assignment such that at most $|A'| - (5n + 2m + 4)$ clauses of $\varphi$ are unsatisfied, that is, $m - m'_{\mathrm{s}} \le |A'| - (5n + 2m + 4)$.
\end{claim}
\begin{proof}
By the minimality of $A'$, both $A' \cap A_1$ and $A' \cap A_2$ contain unique $s$-$t$ dipaths, $P_1$ and $P_2$, respectively.
Moreover, $A' = P_1 \cup P_2$.
Clearly, $P_1$ contains either $P^+_i$ or $P^-_i$ for each $i \in \{ 1, \ldots, n \}$.
Set the value of $x_i$ to be true in the former case, and set its value to be false in the latter case.
By construction, for each $j \in \{ 1, \ldots, m \}$, this assignment does not satisfy the clause $C_j$ if and only if the sub-dipath of $P_2$ from $c_j$ to $c_{j+1}$ does not intersect $P_1$.
Observe that $P_1$ contains exactly $5n+2$ arcs, among which exactly $2n$ is in $A_1 \cap A_2$.
By construction, $P_2$ contains the vertices $c_1, \ldots, c_{m+1}$ in this order with exactly three arcs between them.
Thus $P_2$ contains exactly $3m+2$ arcs, among which exactly $m$ is in $A_1 \cap A_2$.
These imply that
\begin{multline*}
 m - m_{\mathrm{s}}' = \big| (P_2 \setminus P_1) \cap (A_1 \cap A_2) \big| = \big| A' \cap (A_1 \cap A_2) \big| - \big| P_1 \cap (A_1 \cap A_2) \big| \\
 = |A'| - \big| A' \setminus (A_1\cap A_2) \big| - 2n
 = |A'| - \big( (5n+2 - 2n) + (3m+2 - m) \big) - 2n = |A'| - (5n + 2m + 4),
\end{multline*}
which proves the claim.
\end{proof}

Finally, we show that for any $\varepsilon > 0$ and for any inclusion-wise feasible solution $A'$ of the \SupDirSPMinCost with $|A'| \le (1 + \varepsilon) \cdot |A^*|$, the truth assignment constructed in the proof of Claim~\ref{cl:SatToSdsp} satisfies at least $(1 - 17 \varepsilon) \cdot m_{\mathrm{s}}^*$ clauses of $\varphi$.

Note that by the above calculations, we have
\begin{multline*}
 m'_{\mathrm{s}} \ge m - \big( |A'| - (5n + 2m + 4) \big) \ge m - \big( (1 + \varepsilon) \cdot |A^*| - (5n + 2m + 4) \big) \\
 = m - (1 + \varepsilon) \big( |A^*| - (5n + 2m + 4) \big) - \varepsilon \cdot (5n + 2m + 4) \ge m - (1 + \varepsilon)(m - m_{\mathrm{s}}^*) - \varepsilon \cdot (5n + 2m + 4) \\
 = (1 + \varepsilon) \cdot m_{\mathrm{s}}^* - \varepsilon \cdot (5n + 3m + 4) \text{.}
\end{multline*}

Since each variable occurs at least twice and each clause contains at most 2 literals, clearly $2n \le 2m$ holds.
Without loss of generality, we can assume that $m \ge 4$.
Note that a uniform random truth assignment satisfies at least $m/2$ clauses in expectation as clauses of size one are satisfied with probability $1/2$ and those of size two with probability $3/4$.
Thus $m \le 2 m_{\mathrm{s}}^*$.
Therefore, we can obtain
\begin{multline*}
 m'_{\mathrm{s}} \ge (1 + \varepsilon) \cdot m_{\mathrm{s}}^* - \varepsilon \cdot (5n + 3m + 4) \ge (1 + \varepsilon) \cdot m_{\mathrm{s}}^* - \varepsilon \cdot (5m + 3m + m) \ge (1 + \varepsilon) \cdot m_{\mathrm{s}}^* - \varepsilon \cdot 9 m \\
 \ge (1 + \varepsilon) \cdot m_{\mathrm{s}}^* - \varepsilon \cdot 18 m_{\mathrm{s}}^* \ge \left( 1 - 17 \varepsilon \right) \cdot m_{\mathrm{s}}^* \text{.}
\end{multline*}

Hence if $\varepsilon < 1/34204$, then given a ($1+\varepsilon$)-approximation algorithm for \SupDirSPMinCost, one obtains a ($1-17\varepsilon$)-approximation algorithm for MAX-2SAT3, which is better than the inapproximability threshold.
This completes the proof of the theorem.
\end{proof}

In the rest of this section, we focus on directed acyclic graphs.

\begin{thm}\label{thm:exact_DAG_NPC}
The \ExDirSPEx problem is NP-complete even for directed acyclic graphs.
\end{thm}
\begin{proof}
The problem is clearly in NP.
We give a reduction from the NP-complete 3SAT3 problem~\cite{Tovey1984Simplified} to the \ExDirSPEx problem for directed acyclic graphs.
Our goal in the 3SAT3 problem is to decide whether a Boolean formula in conjunctive normal form where each clause has size at most 3 and each variable occurs at most 3 times is satisfiable.

Let $\varphi = C_1 \wedge \cdots \wedge C_m$ be an instance of 3SAT3, where $C_1,\ldots,C_m$ are the clauses of $\varphi$, and let $x_1,\ldots,x_n$ denote its variables.
Without loss of generality, we assume that every variable appears both negated and non-negated, and also that each variable appears at most once per clause.

We define a directed acyclic graph $D=(V,A)$ as follows.
For each $i \in \{ 1, \ldots, n \}$, take the distinct vertices $u_{i1}, \ldots, u_{i4}, v_{i1}, \ldots, v_{i4}$, and $w_i$, then add the arcs of the dipaths  $P_i^+ = (w_i, v_{i1}, \ldots, w_{14},\allowbreak v_{i4}, w_{i+1})$ and $P_i^- = (w_i, u_{i1}, \ldots, u_{i4}, w_{i+1})$ associated to the literals $x_i$ and $\neg x_i$, respectively.
For each $j \in \{ 1, \ldots, m \}$, add two new vertices $s_j$ and $t_j$ associated to the clause $C_j$.
For each $i \in \{ 1, \ldots, n \}$, consider the occurrences of the variable $x_i$ in $\varphi$.
Let $j_1, j_2 \in \{ 1, \ldots, m \}$ be the smallest indices for which $x_i$ is non-negated in the clause $C_{j_1}$ and is negated in $C_{j_2}$, respectively.
Then add the arcs $s_{j_1} v_{i1}$, $v_{i2} t_{j_1}$ and $s_{j_2} u_{i1}$, $u_{i2} t_{j_2}$ to the digraph.
If $x_i$ occurs 3 times, then let $j_3 \in \{ 1, \ldots, m \}$ denote the index of the clause other than $C_{j_1}$ and $C_{j_2}$ in which $x_i$ also appears.
If $x_i$ appears non-negated in $C_{j_3}$, then add the arcs $s_{j_3} v_{i3}$, $v_{i4} t_{j_3}$ to the digraph; if $x_i$ appears negated in $C_{j_3}$, then add the arcs $s_{j_3} u_{i3}$, $u_{i4} t_{j_3}$ to the digraph.
Add two new distinct vertices $s$ and $t$, and add the arcs $sw_1$ and $w_{n+1} t$, and for each $j \in \{ 1, \ldots, m \}$, add the arcs $ss_j$ and $t_j t$ to the digraph.

Let $A_j$ consist of the arcs $ss_j, t_j t$, and the arcs of all $s_j$-$t_j$ dipaths of length 3 for each $j \in \{ 1, \ldots, m \}$.
Let $A_{m+1}$ consist of the arcs $sw_1, w_{n+1} t$ and the arcs of the dipaths $P^+_i$ and $P^-_i$ for each $i \in \{ 1, \ldots, n \}$.
For an example, see Figure~\ref{fig:construction_ExDagSP}.
Clearly, the new instance has polynomial size in the number of clauses and variables of the Boolean formula.

\begin{figure}[!ht]
\centering
\begin{tikzpicture}[scale=0.625]
 \tikzstyle{vertex}=[draw,circle,fill,minimum size=6,inner sep=0]

 \tikzset{edge_blue/.style={postaction={decorate, decoration={markings, mark=between positions 3pt and 1-10pt step 6pt with {\draw[blue!50!black,thin,fill=blue!50] (0:0.08) -- (120:0.08) -- (240:0.08) -- (0:0.08);}}}}}

 \tikzset{edge_red/.style={postaction={decorate, decoration={markings, mark=between positions 3pt and 1-8pt step 7pt with {\draw[red!50!black,thin,fill=red!50] (30:0.08) -- (150:0.08) -- (210:0.08) -- (330:0.08) -- (30:0.08);}}}}}

 \tikzset{edge_green/.style={postaction={decorate, decoration={markings, mark=between positions 5pt and 1-8pt step 8.5pt with {\draw[green!50!black,thin,fill=green!50] (0:0.16) -- (90:0.08) -- (180:0.08) -- (270:0.08) -- (0:0.16);}}}}}

 \tikzset{edge_yellow/.style={postaction={decorate, decoration={markings, mark=between positions 5pt and 1-8pt step 7pt with {\draw[yellow!50!black,thin,fill=yellow!50] (0:0.09) -- (72:0.09) -- (144:0.09) -- (216:0.09) -- (288:0.09) -- (0:0.09);}}}}}

 \tikzset{edge_red_and_blue/.style={postaction={decorate, decoration={markings, mark=between positions 3pt and 1-8pt step 11pt with {\draw[red!50!black,thin,fill=red!50] (30:0.08) -- (150:0.08) -- (210:0.08) -- (330:0.08) -- (30:0.08);}}}, postaction={decorate, decoration={markings, mark=between positions 8.5pt and 1-8pt step 11pt with {\draw[blue!50!black,thin,fill=blue!50] (0:0.08) -- (120:0.08) -- (240:0.08) -- (0:0.08);}}}}}

 \tikzset{edge_red_and_green/.style={postaction={decorate, decoration={markings, mark=between positions 3pt and 1-8pt step 13pt with {\draw[red!50!black,thin,fill=red!50] (30:0.08) -- (150:0.08) -- (210:0.08) -- (330:0.08) -- (30:0.08);}}}, postaction={decorate, decoration={markings, mark=between positions 8.5pt and 1-8pt step 13pt with {\draw[green!50!black,thin,fill=green!50] (0:0.16) -- (90:0.08) -- (180:0.08) -- (270:0.08) -- (0:0.16);}}}}}

 \tikzset{edge_red_and_yellow/.style={postaction={decorate, decoration={markings, mark=between positions 3pt and 1-8pt step 13pt with {\draw[red!50!black,thin,fill=red!50] (30:0.08) -- (150:0.08) -- (210:0.08) -- (330:0.08) -- (30:0.08);}}}, postaction={decorate, decoration={markings, mark=between positions 9.5pt and 1-8pt step 13pt with {\draw[yellow!50!black,thin,fill=yellow!50] (0:0.09) -- (72:0.09) -- (144:0.09) -- (216:0.09) -- (288:0.09) -- (0:0.09);}}}}}

 \node at (-2,4.5) {$P^-_1$};
 \node at (-2,12.5) {$P^-_2$};
 \node at (-2,20.5) {$P^-_3$};
 \node at (2,4.5) {$P^+_1$};
 \node at (2,12.5) {$P^+_2$};
 \node at (2,20.5) {$P^+_3$};

 \node[vertex] (s) at (-7,-2) [label={[xshift=-15pt, yshift=-10pt]:$s$}] {};
 \node[vertex] (t) at (7,26) [label={[xshift=15pt, yshift=-10pt]:$t$}] {};

 \node[vertex] (w1) at (0,0) [label={[xshift=0pt, yshift=-25pt]:$w_1$}] {};
 \node[vertex] (u11) at (-2,1) [label={[xshift=-17pt, yshift=-14pt]:$u_{11}$}] {};
 \node[vertex] (u12) at (-3,3) [label={[xshift=-14pt, yshift=-11pt]:$u_{12}$}] {};
 \node[vertex] (u13) at (-3,5) [label={[xshift=-15pt, yshift=-12pt]:$u_{13}$}] {};
 \node[vertex] (u14) at (-2,7) [label={[xshift=-17pt, yshift=-14pt]:$u_{14}$}] {};
 \node[vertex] (v11) at (2,1) [label={[xshift=10pt, yshift=-22pt]:$v_{11}$}] {};
 \node[vertex] (v12) at (3,3) [label={[xshift=14pt, yshift=-20pt]:$v_{12}$}] {};
 \node[vertex] (v13) at (3,5) [label={[xshift=14pt, yshift=-10pt]:$v_{13}$}] {};
 \node[vertex] (v14) at (2,7) [label={[xshift=14pt, yshift=-10pt]:$v_{14}$}] {};
 \node[vertex] (w2) at (0,8) [label={[xshift=1pt, yshift=-27pt]:$w_2$}] {};
 \node[vertex] (u21) at (-2,9) [label={[xshift=-15pt, yshift=-14pt]:$u_{21}$}] {};
 \node[vertex] (u22) at (-3,11) [label={[xshift=-14pt, yshift=-12pt]:$u_{22}$}] {};
 \node[vertex] (u23) at (-3,13) [label={[xshift=-15pt, yshift=-14pt]:$u_{23}$}] {};
 \node[vertex] (u24) at (-2,15) [label={[xshift=-17pt, yshift=-14pt]:$u_{24}$}] {};
 \node[vertex] (v21) at (2,9) [label={[xshift=15pt, yshift=-14pt]:$v_{21}$}] {};
 \node[vertex] (v22) at (3,11) [label={[xshift=14pt, yshift=-6pt]:$v_{22}$}] {};
 \node[vertex] (v23) at (3,13) [label={[xshift=14pt, yshift=-10pt]:$v_{23}$}] {};
 \node[vertex] (v24) at (2,15) [label={[xshift=14pt, yshift=-10pt]:$v_{24}$}] {};
 \node[vertex] (w3) at (0,16) [label={[xshift=0pt, yshift=-27pt]:$w_3$}] {};
 \node[vertex] (u31) at (-2,17) [label={[xshift=-14pt, yshift=-18pt]:$u_{31}$}] {};
 \node[vertex] (u32) at (-3,19) [label={[xshift=-15pt, yshift=-14pt]:$u_{32}$}] {};
 \node[vertex] (u33) at (-3,21) [label={[xshift=-15pt, yshift=-14pt]:$u_{33}$}] {};
 \node[vertex] (u34) at (-2,23) [label={[xshift=-15pt, yshift=-14pt]:$u_{34}$}] {};
 \node[vertex] (v31) at (2,17) [label={[xshift=16pt, yshift=-14pt]:$v_{31}$}] {};
 \node[vertex] (v32) at (3,19) [label={[xshift=15pt, yshift=-5pt]:$v_{32}$}] {};
 \node[vertex] (v33) at (3,21) [label={[xshift=17pt, yshift=-14pt]:$v_{33}$}] {};
 \node[vertex] (v34) at (2,23) [label={[xshift=17pt, yshift=-14pt]:$v_{34}$}] {};
 \node[vertex] (w4) at (0,24) [label={[xshift=0pt, yshift=-1pt]:$w_4$}] {};

 \node[vertex] (s1) at (-10,4) [label={[xshift=-17pt, yshift=-14pt]:$s_1$}] {};
 \node[vertex] (s2) at (-10,12) [label={[xshift=-17pt, yshift=-14pt]:$s_2$}] {};
 \node[vertex] (s3) at (-10,20) [label={[xshift=-17pt, yshift=-14pt]:$s_3$}] {};
 \node[vertex] (t1) at (10,4) [label={[xshift=17pt, yshift=-14pt]:$t_1$}] {};
 \node[vertex] (t2) at (10,12) [label={[xshift=17pt, yshift=-14pt]:$t_2$}] {};
 \node[vertex] (t3) at (10,20) [label={[xshift=17pt, yshift=-14pt]:$t_3$}] {};

 \begin{pgfonlayer}{background}
      \begin{scope}[opacity=.8,transparency group]
        \highlight{10pt}{black!20}{(s.center) to [bend left =0] (w1.center) to [bend left =0] (v11.center) to [bend left =0] (v12.center) to [bend left =0] (v13.center) to [bend left =0] (v14.center) to [bend left =0] (w2.center) to [bend left =0] (u21.center) to [bend left =0] (u22.center) to [bend left =0] (u23.center) to [bend left =0] (u24.center) to [bend left =0] (w3.center) to [bend left =0] (v31.center) to [bend left =0] (v32.center) to [bend left =0] (v33.center) to [bend left =0] (v34.center) to [bend left =0] (w4.center) to [bend left =0] (t.center)}
        \highlight{10pt}{black!20}{(s.center) to [bend left =0] (s1.center) to [bend left =0] (v11.center) to [bend left =0] (v12.center) to [bend right =0] (t1.center) to [bend left =10] (t.center)}
        \highlight{10pt}{black!20}{(s.center) to [bend left =8] (s2.center) to [bend left =10] (v31.center) to [bend left =0] (v32.center) to [bend right =0] (t2.center) to [bend left =0] (t.center)}
        \highlight{10pt}{black!20}{(s.center) to [bend left=0] (s3.center) to [bend right =20] (u21.center) to [bend left =0] (u22.center) to [bend right =10] (t3.center) to [bend left =0] (t.center)}
      \end{scope}
 \end{pgfonlayer}

 \path[edge_green] (s2) to [bend right=20] (u11);
 \draw[->,>={LaTeX[black,length=9pt, width=9pt]},draw opacity=0] (s2) to [bend right=20] (u11);

 \path[edge_red] (s) -- (w1);
 \draw[->,>={LaTeX[black,length=9pt, width=9pt]},draw opacity=0] (s) -- (w1);
 \path[edge_red] (w4) -- (t);
 \draw[->,>={LaTeX[black,length=9pt, width=9pt]},draw opacity=0] (w4) -- (t);
 \path[edge_green] (s) to [bend left=8] (s2);
 \draw[->,>={LaTeX[black,length=9pt, width=9pt]},draw opacity=0] (s) to [bend left=8] (s2);
 \path[edge_blue] (s) -- (s1);
 \draw[->,>={LaTeX[black,length=9pt, width=9pt]},draw opacity=0] (s) -- (s1);
 \path[edge_blue] (t1) to [bend left=10] (t);
 \draw[->,>={LaTeX[black,length=9pt, width=9pt]},draw opacity=0] (t1) to [bend left=10] (t);
 \path[edge_green] (t2) -- (t);
 \draw[->,>={LaTeX[black,length=9pt, width=9pt]},draw opacity=0] (t2) -- (t);
 \path[edge_yellow] (s) -- (s3);
 \draw[->,>={LaTeX[black,length=9pt, width=9pt]},draw opacity=0] (s) -- (s3);
 \path[edge_yellow] (t3) -- (t);
 \draw[->,>={LaTeX[black,length=9pt, width=9pt]},draw opacity=0] (t3) -- (t);

 \path[edge_blue] (s1) -- (v11);
 \draw[->,>={LaTeX[black,length=9pt, width=9pt]},draw opacity=0] (s1) -- (v11);
 \path[edge_blue] (v12) -- (t1);
 \draw[->,>={LaTeX[black,length=9pt, width=9pt]},draw opacity=0] (v12) -- (t1);
 \path[edge_blue] (s1) to [bend left=10] (v21);
 \draw[->,>={LaTeX[black,length=9pt, width=9pt]},draw opacity=0] (s1) to [bend left=10] (v21);
 \path[edge_blue] (v22) -- (t1);
 \draw[->,>={LaTeX[black,length=9pt, width=9pt]},draw opacity=0] (v22) -- (t1);

 \path[edge_green] (u12) to [bend right=10] (t2);
 \draw[->,>={LaTeX[black,length=9pt, width=9pt]},draw opacity=0] (u12) to [bend right=10] (t2);
 \path[edge_green] (s2) to [bend left=10] (v31);
 \draw[->,>={LaTeX[black,length=9pt, width=9pt]},draw opacity=0] (s2) to [bend left=10] (v31);
 \path[edge_green] (v32) -- (t2);
 \draw[->,>={LaTeX[black,length=9pt, width=9pt]},draw opacity=0] (v32) -- (t2);

 \path[edge_yellow] (s3) to [bend right=20] (u21);
 \draw[->,>={LaTeX[black,length=9pt, width=9pt]},draw opacity=0] (s3) to [bend right=20] (u21);
 \path[edge_yellow] (u22) to [bend right=10] (t3);
 \draw[->,>={LaTeX[black,length=9pt, width=9pt]},draw opacity=0] (u22) to [bend right=10] (t3);
 \path[edge_yellow] (s3) to [bend right=10] (u31);
 \draw[->,>={LaTeX[black,length=9pt, width=9pt]},draw opacity=0] (s3) to [bend right=10] (u31);
 \path[edge_yellow] (u32) to [bend left=5] (t3);
 \draw[->,>={LaTeX[black,length=9pt, width=9pt]},draw opacity=0] (u32) to [bend left=5] (t3);

 \path[edge_red] (w1) -- (u11);
 \draw[->,>={LaTeX[black,length=9pt, width=9pt]},draw opacity=0] (w1) -- (u11);
 \path[edge_red_and_green] (u11) -- (u12);
 \draw[->,>={LaTeX[black,length=9pt, width=9pt]},draw opacity=0] (u11) -- (u12);
 \path[edge_red] (u12) -- (u13);
 \draw[->,>={LaTeX[black,length=9pt, width=9pt]},draw opacity=0] (u12) -- (u13);
 \path[edge_red] (u13) -- (u14);
 \draw[->,>={LaTeX[black,length=9pt, width=9pt]},draw opacity=0] (u13) -- (u14);
 \path[edge_red] (u14) -- (w2);
 \draw[->,>={LaTeX[black,length=9pt, width=9pt]},draw opacity=0] (u14) -- (w2);
 \path[edge_red] (w1) -- (v11);
 \draw[->,>={LaTeX[black,length=9pt, width=9pt]},draw opacity=0] (w1) -- (v11);
 \path[edge_red_and_blue] (v11) -- (v12);
 \draw[->,>={LaTeX[black,length=9pt, width=9pt]},draw opacity=0] (v11) -- (v12);
 \path[edge_red] (v12) -- (v13);
 \draw[->,>={LaTeX[black,length=9pt, width=9pt]},draw opacity=0] (v12) -- (v13);
 \path[edge_red] (v13) -- (v14);
 \draw[->,>={LaTeX[black,length=9pt, width=9pt]},draw opacity=0] (v13) -- (v14);
 \path[edge_red] (v14) -- (w2);
 \draw[->,>={LaTeX[black,length=9pt, width=9pt]},draw opacity=0] (v14) -- (w2);
 \path[edge_red] (w2) -- (u21);
 \draw[->,>={LaTeX[black,length=9pt, width=9pt]},draw opacity=0] (w2) -- (u21);
 \path[edge_red_and_yellow] (u21) -- (u22);
 \draw[->,>={LaTeX[black,length=9pt, width=9pt]},draw opacity=0] (u21) -- (u22);
 \path[edge_red] (u22) -- (u23);
 \draw[->,>={LaTeX[black,length=9pt, width=9pt]},draw opacity=0] (u22) -- (u23);
 \path[edge_red] (u23) -- (u24);
 \draw[->,>={LaTeX[black,length=9pt, width=9pt]},draw opacity=0] (u23) -- (u24);
 \path[edge_red] (u24) -- (w3);
 \draw[->,>={LaTeX[black,length=9pt, width=9pt]},draw opacity=0] (u24) -- (w3);
 \path[edge_red] (w2) -- (v21);
 \draw[->,>={LaTeX[black,length=9pt, width=9pt]},draw opacity=0] (w2) -- (v21);
 \path[edge_red_and_blue] (v21) -- (v22);
 \draw[->,>={LaTeX[black,length=9pt, width=9pt]},draw opacity=0] (v21) -- (v22);
 \path[edge_red] (v22) -- (v23);
 \draw[->,>={LaTeX[black,length=9pt, width=9pt]},draw opacity=0] (v22) -- (v23);
 \path[edge_red] (v23) -- (v24);
 \draw[->,>={LaTeX[black,length=9pt, width=9pt]},draw opacity=0] (v23) -- (v24);
 \path[edge_red] (v24) -- (w3);
 \draw[->,>={LaTeX[black,length=9pt, width=9pt]},draw opacity=0] (v24) -- (w3);
 \path[edge_red] (w3) -- (u31);
 \draw[->,>={LaTeX[black,length=9pt, width=9pt]},draw opacity=0] (w3) -- (u31);
 \path[edge_red_and_yellow] (u31) -- (u32);
 \draw[->,>={LaTeX[black,length=9pt, width=9pt]},draw opacity=0] (u31) -- (u32);
 \path[edge_red] (u32) -- (u33);
 \draw[->,>={LaTeX[black,length=9pt, width=9pt]},draw opacity=0] (u32) -- (u33);
 \path[edge_red] (u33) -- (u34);
 \draw[->,>={LaTeX[black,length=9pt, width=9pt]},draw opacity=0] (u33) -- (u34);
 \path[edge_red] (u34) -- (w4);
 \draw[->,>={LaTeX[black,length=9pt, width=9pt]},draw opacity=0] (u34) -- (w4);
 \path[edge_red] (w3) -- (v31);
 \draw[->,>={LaTeX[black,length=9pt, width=9pt]},draw opacity=0] (w3) -- (v31);
 \path[edge_red_and_green] (v31) -- (v32);
 \draw[->,>={LaTeX[black,length=9pt, width=9pt]},draw opacity=0] (v31) -- (v32);
 \path[edge_red] (v32) -- (v33);
 \draw[->,>={LaTeX[black,length=9pt, width=9pt]},draw opacity=0] (v32) -- (v33);
 \path[edge_red] (v33) -- (v34);
 \draw[->,>={LaTeX[black,length=9pt, width=9pt]},draw opacity=0] (v33) -- (v34);
 \path[edge_red] (v34) -- (w4);
 \draw[->,>={LaTeX[black,length=9pt, width=9pt]},draw opacity=0] (v34) -- (w4);
\end{tikzpicture}
\caption{The directed acyclic graph constructed in the proof of Theorem~\ref{thm:exact_DAG_NPC} to the Boolean formula $\varphi = (x_1 \lor x_2) \land (\neg x_1 \lor x_3) \land (\neg x_2 \lor \neg x_3)$.
  The arcs marked with (blue) triangles, with (green) deltoids, with (yellow) pentagons, and with (red) rectangles correspond to the arcs in $A_1$, $A_2$, $A_3$, and $A_4$, respectively.
  The arcs with alternating markings belong to the corresponding intersections of these arc sets.
  The highlighted arcs form the subgraph $A'$ corresponding to the truth assignment $x_1 = x_3 = \texttt{true}$ and $x_2 = \texttt{false}$.}
\label{fig:construction_ExDagSP}
\end{figure}
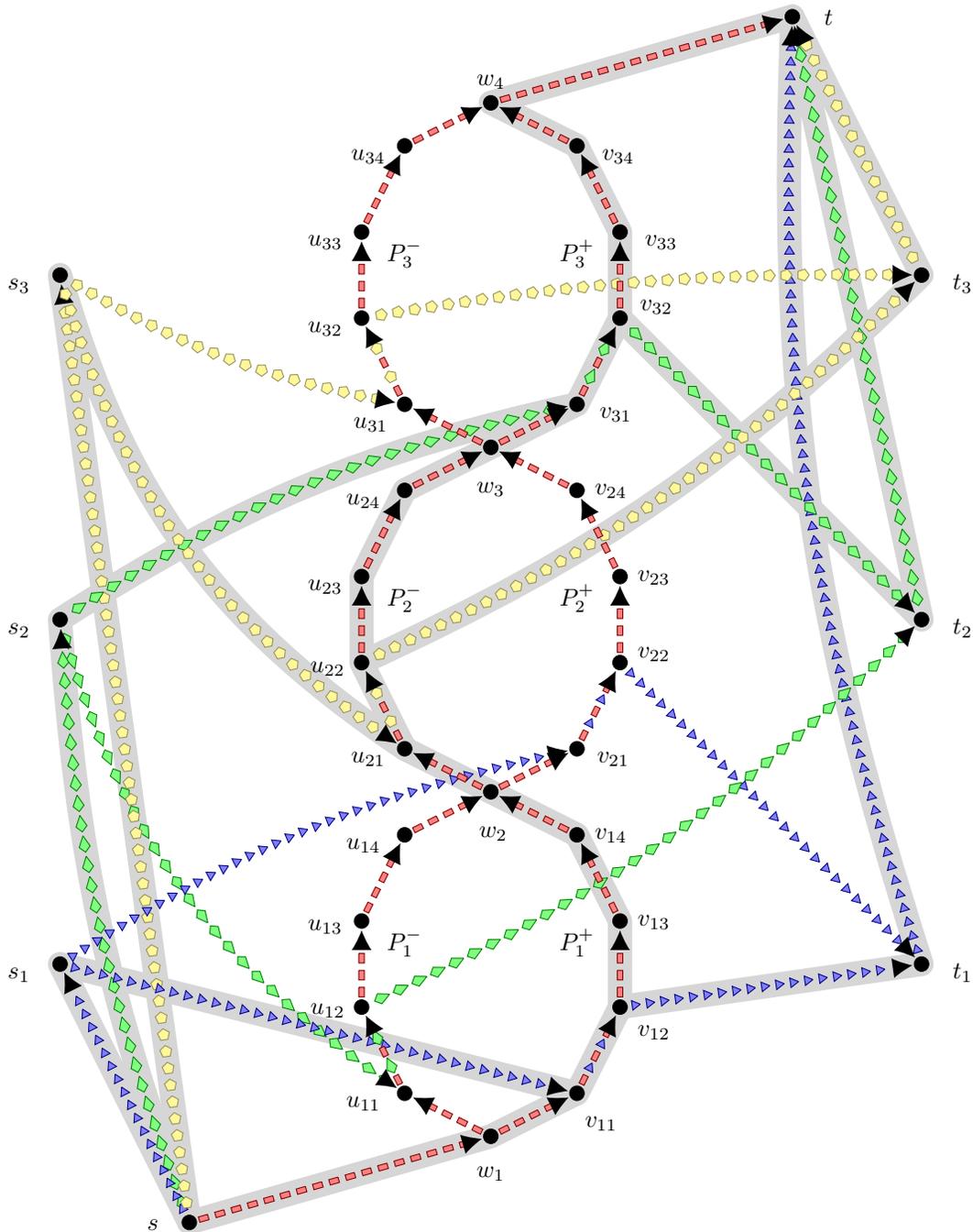

First, assume that $\varphi$ is satisfiable, and consider a truth assignment that satisfies $\varphi$.
For each $j \in \{ 1, \ldots, m \}$, let us define an index $i_j \in \{ 1, \ldots, n \}$ as follows: let it be the index of the variable corresponding to the first true literal occurring in $C_j$.
Now we construct an arc set $A'$.
Let $A'$ contain all the arcs incident to $s$ or $t$.
For each $i \in \{ 1, \ldots, n \}$, add the arcs of $P_i^+$ to $A'$ if $x_i$ is true, otherwise add the arcs of $P_i^-$ to $A'$.
For each $j \in \{ 1, \ldots, m \}$, add the unique arc going from the vertex $s_j$ to the vertices of $P_{i_j}^+ \cup P_{i_j}^-$, and add the unique arc going from the vertices of $P_{i_j}^+ \cup P_{i_j}^-$ to the vertex $t_j$ to $A'$.
It is not difficult to see that $A'$ is a feasible solution of \ExDirSPEx.

Now assume that there exists a feasible solution $A' \subseteq A$ of \ExDirSPEx in $D$.
Clearly, any $s$-$t$ dipath in $A_{m+1}$ contains either $P^+_i$ or $P^-_i$ for each $i \in \{ 1, \ldots, n \}$.
Set the value of $x_i$ to be true in the former case, and set its value to be false in the latter case.
By construction, this assignment satisfies each clause of $\varphi$.
Thus $\varphi$ is satisfiable.
\end{proof}

\begin{thm}\label{thm:superset_DAG_APXC}
The \SupDirSPMinCost problem does not have a polynomial-time $\big( (1 - \varepsilon) \ln k \big)$-approximation algorithm for any $\varepsilon > 0$, even for directed acyclic graphs with $c \equiv 1$, unless $\text{P} = \text{NP}$.
\end{thm}
\begin{proof}
We give a polynomial-time reduction from the \textsc{SetCover} problem to \SupDirSPMinCost.
The input of \textsc{SetCover} is a finite ground set $U$ and a family $\mathcal{F}$ of some subsets of $U$, and our goal is to find a subfamily $\mathcal{C} \subseteq \mathcal{F}$ of minimum cardinality such that $\bigcup_{C \in \mathcal{C}} C = U$.
It was shown in~\cite{Dinur2014Analytical} that \textsc{SetCover} is NP-hard to approximate within factor $(1 - \varepsilon) \ln |U|$ for any $\varepsilon > 0$.
Now we show that if there were a polynomial-time $\big( (1 - \varepsilon) \ln k \big)$-approximation algorithm for the \SupDirSPMinCost problem for some $\varepsilon > 0$, then we would obtain a polynomial-time $\big( (1 - \varepsilon) \ln |U| \big)$-approximation algorithm for \textsc{SetCover}, implying $\text{P} = \text{NP}$.

Let $(U, \mathcal{F})$ be an instance of the \textsc{SetCover} problem.
Let $u_1, \ldots, u_k$ denote the elements of $U$, i.e.\ $U = \{ u_1, \ldots, u_k \}$.
We define a directed acyclic graph $D=(V,A)$ as follows.
Take two distinct vertices $s$ and $t$.
For each $F \in \mathcal{F}$, add an arc $a_F = st$ to the digraph.
Let $A_i = \{ a_F : F \in \mathcal{F}, u_i \in F \}$ for any $u_i \in U$.
For an example, see Figure~\ref{fig:construction_SupDagSP}.

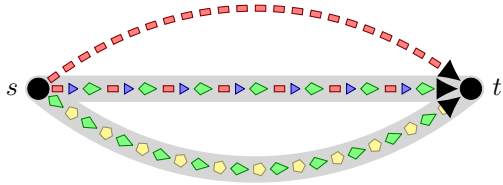
\begin{figure}[!ht]
\centering
\begin{tikzpicture}[scale=1.15]
 \tikzstyle{vertex}=[draw,circle,fill,minimum size=8,inner sep=0]

 \tikzset{edge_blue/.style={postaction={decorate, decoration={markings, mark=between positions 3pt and 1-10pt step 6pt with {\draw[blue!50!black,thin,fill=blue!50] (0:0.08) -- (120:0.08) -- (240:0.08) -- (0:0.08);}}}}}

 \tikzset{edge_red/.style={postaction={decorate, decoration={markings, mark=between positions 3pt and 1-8pt step 7pt with {\draw[red!50!black,thin,fill=red!50] (30:0.08) -- (150:0.08) -- (210:0.08) -- (330:0.08) -- (30:0.08);}}}}}

 \tikzset{edge_green/.style={postaction={decorate, decoration={markings, mark=between positions 5pt and 1-8pt step 8.5pt with {\draw[green!50!black,thin,fill=green!50] (0:0.16) -- (90:0.08) -- (180:0.08) -- (270:0.08) -- (0:0.16);}}}}}

 \tikzset{edge_yellow/.style={postaction={decorate, decoration={markings, mark=between positions 5pt and 1-8pt step 7pt with {\draw[yellow!50!black,thin,fill=yellow!50] (0:0.09) -- (72:0.09) -- (144:0.09) -- (216:0.09) -- (288:0.09) -- (0:0.09);}}}}}

 \tikzset{edge_red_and_blue_and_green/.style={postaction={decorate, decoration={markings, mark=between positions 3pt and 1-8pt step 21pt with {\draw[red!50!black,thin,fill=red!50] (30:0.08) -- (150:0.08) -- (210:0.08) -- (330:0.08) -- (30:0.08);}}}, postaction={decorate, decoration={markings, mark=between positions 8.5pt and 1-8pt step 21pt with {\draw[blue!50!black,thin,fill=blue!50] (0:0.08) -- (120:0.08) -- (240:0.08) -- (0:0.08);}}}, postaction={decorate, decoration={markings, mark=between positions 15pt and 1-8pt step 21pt with {\draw[green!50!black,thin,fill=green!50] (0:0.16) -- (90:0.08) -- (180:0.08) -- (270:0.08) -- (0:0.16);}}}}}

 \tikzset{edge_green_and_yellow/.style={postaction={decorate, decoration={markings, mark=between positions 3pt and 1-8pt step 15pt with {\draw[green!50!black,thin,fill=green!50] (0:0.16) -- (90:0.08) -- (180:0.08) -- (270:0.08) -- (0:0.16);}}}, postaction={decorate, decoration={markings, mark=between positions 11.5pt and 1-8pt step 15pt with {\draw[yellow!50!black,thin,fill=yellow!50] (0:0.09) -- (72:0.09) -- (144:0.09) -- (216:0.09) -- (288:0.09) -- (0:0.09);}}}}}

 \node[vertex] (s) at (0,0) [label={[xshift=-10pt, yshift=-10pt]:$s$}] {};
 \node[vertex] (t) at (5,0) [label={[xshift=10pt, yshift=-10pt]:$t$}] {};

 \begin{pgfonlayer}{background}
      \begin{scope}[opacity=.8,transparency group]
        \highlight{10pt}{black!20}{(s.center) -- (t.center)}
        \highlight{10pt}{black!20}{(s.center) to [bend right = 39.675] (t.center)}
      \end{scope}
 \end{pgfonlayer}

 \path[edge_red] (s) to [bend left = 37.5] (t);
 \draw[->,>={LaTeX[black,length=9pt, width=9pt]},draw opacity=0] (s) to [bend left = 37.5] (t);

 \path[edge_red_and_blue_and_green] (s) -- (t);
 \draw[->,>={LaTeX[black,length=9pt, width=9pt]},draw opacity=0] (s) -- (t);

 \path[edge_green_and_yellow] (s) to [bend right = 37.5] (t);
 \draw[->,>={LaTeX[black,length=9pt, width=9pt]},draw opacity=0] (s) to [bend right = 37.5] (t);
\end{tikzpicture}
\caption{The directed acyclic graph constructed in the proof of Theorem~\ref{thm:superset_DAG_APXC} to the family $\mathcal{S} = \big\{ \{ u_1 \}, \{ u_1, u_2, u_3 \}, \{ u_3, u_4 \} \big\}$.
The colors (red) rectangles, (blue) triangles, (green) deltoids, and (yellow) pentagons correspond to the arcs in $A_1$, $A_2$, $A_3$, and $A_4$, respectively.
Thus, for example, the middle arc, marked with alternating (red) rectangles, (blue) triangles, and with (green) deltoids, belongs to $A_1 \cap A_2 \cap A_3$.
The highlighted arcs form the subgraph $A'$ corresponding to the set cover $\mathcal{C} = \big\{ \{ u_1, u_2, u_3 \}, \{ u_3, u_4 \} \big\}$.}
\label{fig:construction_SupDagSP}
\end{figure}

Let $m^*$ denote the minimum number of sets in $\mathcal{F}$ that cover all the elements of $U$, and let $A^* \subseteq A$ be an optimal solution of \SupDirSPMinCost in $D$.

Now we show that $m^* \ge |A^*|$.
Consider a minimum set cover $\mathcal{C}^* \subseteq \mathcal{F}$ of $U$, and let $A' = \{ a_F : F \in \mathcal{C}^* \}$.
Clearly, $A'$ is a feasible solution of \SupDirSPMinCost, and $m^* = |\mathcal{C}^*| = |A'| \ge |A^*|$.

Next, let $A'$ be an arbitrary feasible solution of \SupDirSPMinCost in the constructed directed acyclic graph $D$.
Then $\mathcal{C}' = \big\{ F \in \mathcal{F} : a_F \in A' \big\}$ is a set cover with $|\mathcal{C}'| = |A'|$.

We are ready to show that for any $\varepsilon > 0$ and for any feasible solution $A'$ of \SupDirSPMinCost with $|A'| \le \big( (1 - \varepsilon) \ln k \big) |A^*|$, the set cover $\mathcal{C}' = \big\{ F \in \mathcal{F} : a_F \in A' \big\}$ has size at most $\big( (1 - \varepsilon) \ln |U| \big) m^*$.

By the above calculations, we have
\[ |\mathcal{C}'| = |A'| \le \big( (1 - \varepsilon) \ln k \big) |A^*| \le \big( (1 - \varepsilon) \ln |U| \big) m^*.\]

Hence we obtained a $\big( (1 - \varepsilon) \ln |U| \big)$-approximation algorithm for \textsc{SetCover}, which is better than the inapproximability threshold.
This completes the proof of the theorem.
\end{proof}

\subsection{Undirected variants}

Similarly to Theorem~\ref{thm:exact_directed_NPC}, now we show that the following decision version of the {\ExUndSPMinCost} problem in which one is to decide whether a feasible solution exists is NP-complete.

\smallskip

\noindent {\bf \ExUndSPEx} \\
\textit{Input:} a positive integer $k$, a graph~$G=(V,E)$ with $E = E_1 \cup \cdots \cup E_k$, and two distinct vertices $s, t \in V$.\\
\textit{Problem:} decide whether there exist $E' \subseteq E$ such that $E' \cap E_i$ is an $s$-$t$ path for all $i \in \{ 1, \ldots, k \}$.

\medskip

\begin{thm} \label{thm:exact_undirected_NPC}
The {\ExUndSPEx} problem is NP-complete even for $k=2$.
\end{thm}
\begin{proof}
Clearly, the problem is in NP.
We give a reduction from the NP-complete 3SAT3 problem~\cite{Tovey1984Simplified} to \ExUndSPEx.
Our goal in the 3SAT3 problem is to decide whether a Boolean formula in conjunctive normal form where each clause has size at most 3 and each variable occurs at most 3 times is satisfiable.

Let $\varphi = C_1 \wedge \cdots \wedge C_m$ be an instance of 3SAT3, where $C_1,\ldots,C_m$ are the clauses of $\varphi$, and let $x_1,\ldots,x_n$ denote its variables.
Without loss of generality, we assume that every variable appears both negated and non-negated, and also that each variable appears at most once per clause.

Consider the graph $G=(V,E)$ obtained from the digraph $D$ constructed in the proof of Theorem~\ref{thm:superset_directed_APXC} by forgetting its orientation.
Note that the construction of $D$ given in Theorem~\ref{thm:superset_directed_APXC} works regardless of the sizes of the clauses.

First, assume that $\varphi$ is satisfiable.
Then note that the edge set corresponding to $A'$ constructed in the proof of Claim~\ref{cl:SdspToSat} is a feasible solution to {\ExUndSPEx}.

Now assume that there exists $E' \subseteq E$ such that $E' \cap E_1$ and $E' \cap E_2$ both form an $s$-$t$ path.
Then note that the truth assignment constructed to the corresponding arc set $A'$ in the proof of Claim~\ref{cl:SatToSdsp} satisfies~$\varphi$.
\end{proof}

This immediately implies that the optimization counterpart cannot be approximated unless $\text{P} = \text{NP}$.

\begin{thm}\label{thm:exact_undirected_inapprox}
There exists no polynomial-time $\alpha$-approximation algorithm for the \ExUndSPMinCost problem for any $\alpha \geq 1$ unless $\text{P} = \text{NP}$, even if $k=2$.
\end{thm}
\begin{proof}
The proof is analogous to that of Theorem~\ref{thm:exact_directed_inapprox}, relying on Theorem~\ref{thm:exact_undirected_NPC} instead of Theorem~\ref{thm:exact_directed_NPC}.
\end{proof}

Analogously to the case of digraphs, we can obtain a polynomial-time $k$-approximation algorithm for \SupUndSPMinCost by finding a shortest $s$-$t$ path in each color class.
Moreover, we obtain the following inapproximability results.

\begin{thm} \label{thm:super_undirected_APXC}
The \SupUndSPMinCost problem is APX-complete even for $c \equiv 1$ and $k=2$.
\end{thm}
\begin{proof}
The proof is analogous to that of Theorem~\ref{thm:superset_directed_APXC}, where only the orientation of the construction needs to be forgotten.
\end{proof}

\begin{thm}\label{thm:superset_undirected_inapprox}
The \SupUndSPMinCost problem does not have a polynomial-time $\big( (1 - \varepsilon) \ln k \big)$-approximation algorithm for any $\varepsilon > 0$, even for $c \equiv 1$ unless $\text{P} = \text{NP}$.
\end{thm}
\begin{proof}
The proof is analogous to that of Theorem~\ref{thm:superset_DAG_APXC}, where only the orientation of the construction needs to be forgotten.
\end{proof}

\section{Tractable cases}

In this section, we consider tractable special cases of our problems.
First, we present polynomial-time algorithms for the problems \ExDirSPMinCost and the \SupDirSPMinCost in directed acyclic graphs when $k$ is a fixed constant.
We give dynamic programming algorithms, which resemble the approach for the problem of finding $k$ disjoint paths when $k$ is a constant~\cite[page~1244]{Schrijver2003Combinatorial}, which was based on~\cite{Fortune1980Directed}.

Now, we consider the \ExDirSPMinCost problem for directed acyclic graphs.

\begin{thm}\label{thm:exact_dac_constantk}
The {\ExDirSPMinCost} problem can be solved in polynomial time for directed acyclic graphs when $k$ is a constant.
\end{thm}
\begin{proof}
Let $D=(V, A)$ be a directed acyclic graph and $s, t\in V$ be two distinct vertices.
We may assume that $s$ is a source and $t$ is a sink by adding a new vertex with degree one to the original terminal vertices if necessary.
For an arc $xy \in A$, let $I_{xy} = \big\{ i \in \{ 1, \ldots, k \} : xy \in A_i \big\}$, that is, $I_{xy}$ contains an index $i \in \{ 1, \ldots, k \}$ if and only if the arc $xy$ is in the color class $A_i$.

We construct a digraph $\widetilde{D} =(\widetilde{V}, \widetilde{A})$ and a cost function $\widetilde{c} \colon \widetilde{A} \to \mathbb{R}$ as follows.
The vertex set $\widetilde{V}$ consists of all $k$ tuples $(v_1, \ldots, v_k)$ of the vertices of $D$.
For an arc $xy\in A$, add an arc from $(u_1, \ldots, u_k)$ to $(v_1, \ldots, v_k)$ in $\widetilde{D}$ if
\begin{enumerate}[label=(\alph*), topsep=1pt, itemsep=-2pt]
 \item $u_i = x$ and $v_i=y$ for all $i\in I_{xy}$, and
 \item $u_j = v_j$ for all $j\notin I_{xy}$,
\end{enumerate}
and let its $\widetilde{c}$-cost be equal to $c(xy)$.
Clearly, the size of $\widetilde{D}$ is polynomial in that of the input graph, as $k$ is considered a constant.
As shown by the following claim, in this auxiliary digraph, the $(s,\ldots,s)$-$(t,\ldots,t)$ dipaths correspond to the feasible solutions of the \ExDirSPMinCost problem.

\begin{claim}\label{clm:exact_dac_constantk}
For any $\gamma \in \R$, there is a solution $A'\subseteq A$ of \ExDirSPMinCost of $c$-cost $\gamma$ if and only if $\widetilde{D}$ contains a dipath $\widetilde{P}$ from $(s, \ldots, s)$ to $(t, \ldots, t)$ of $\widetilde{c}$-cost $\gamma$.
\end{claim}
\begin{proof}
First, let $A'\subseteq A$ be a solution to \ExDirSPMinCost of $c$-cost~$\gamma$.
Now we construct a dipath $\widetilde{P}$ in $\widetilde{D}$ such that each arc of $\widetilde{P}$ corresponds to an arc of $A'$.
Since each arc in $\widetilde{D}$ was added because of exactly one arc in $D$, this implies that the $\widetilde{c}$-cost of $\widetilde{P}$ is at most~$\gamma$.

For each $i \in \{ 1, \ldots, k \}$, let $P_i=A'\cap A_i$.
Let $\ell_i$ be the length of $P_i$ and assume that the vertices of $P_i$ are $v_{i,0}, \ldots, v_{i,\ell_i}$ in this order, where $v_{i,0}=s$ and $v_{i,\ell_i}=t$.
Since $D$ is a directed acyclic graph, we can sort the vertices of $D$ by a topological order.

Assume that we have already constructed a dipath from $(s, \ldots, s)$ to $(v_{1,h_1}, \ldots, v_{k,h_k})$ for some $h_i \in \{ 0, \ldots, \ell_i \}$ for each $i \in \{ 1, \ldots, k \}$.
If $h_i=\ell_i$ for all $i\in\{1, \ldots, k\}$, then we are done.
Otherwise, take an index $i$ for which $v_{i,h_i}$ is the first vertex in the topological order of $D$ among $v_{1,h_1}, \ldots, v_{k,h_k}$.
Let $x=v_{i,h_i}$ and $y=v_{i,h_i + 1}$.
Note that $xy \in A'$.
Extend the already constructed dipath with the arc in $\widetilde{D}$ that goes from $(v_{1,h_1}, \ldots, v_{k,h_k})$ to $(w_1, \ldots, w_k)$, where $w_j=v_{j,h_j + 1}$ if $j\in I_{xy}$, and $w_j = v_{j,h_j}$ otherwise.
Such an arc exists in $\widetilde{D}$ because the special selection of $i$ ensures that $v_{j,h_j} = x$ for each $j \in I_{xy}$, hence both (a) and (b) hold.
This way, we step $v_{i,h_{i}}$ to $v_{i,h_{i}+1}$, and hence we reach the vertex $(t,\ldots,t)$ in at most $|V|^2$ iterations.

\medskip

Second, let $\widetilde{P}$ be an $(s,\ldots,s)$-$(t,\ldots,t)$ dipath of $c$-cost $\gamma$.
We construct a solution $A'\subseteq A$ to \ExDirSPMinCost of $c$-cost $\gamma$.
Let $\ell$ be the length of $\widetilde{P}$ and assume that the vertices of $\widetilde{P}$ are $(v_{1,j},\ldots, v_{k,j})$ for each $j \in \{ 0,\ldots, \ell \}$ in this order, where $(v_{1,0}, \ldots, v_{k,0}) = (s, \ldots, s)$ and $(v_{1,\ell},\ldots, v_{k,\ell}) = (t, \ldots, t)$.
For each $i \in \{ 1, \ldots, k \}$, consider the vertices $v_{i,j}$ for each $j \in \{ 0, \ldots, \ell \}$ in this order taking repeated vertices only once.
Note that for each $i \in \{ 1, \ldots, k \}$ and $j \in \{ 0, \ldots, \ell \}$, if $v_{i,j} \ne v_{i,j+1}$, then there is an arc from $v_{i,j}$ to $v_{i,j+1}$ in $D$ for which $i \in I_{v_{i,j} v_{i,j+1}}$.
Since $D$ is acyclic, this implies that $P_i$ is a dipath from $s$ to $t$ in the color class $A_i$.
Let $A'$ be the union of the dipath $P_i$ for each $i \in \{1, \ldots, k\}$, which is a feasible solution to \ExDirSPMinCost.
\end{proof}

It follows from Claim~\ref{clm:exact_dac_constantk} that we can find a feasible solution to \ExDirSPMinCost of minimum cost by computing the shortest dipath in $\widetilde{D}$ from $(s, \ldots, s)$ to $(t, \ldots, t)$.
Since $\widetilde{D}$ has $O \big( |V|^k \big)$ vertices and $O \big( |V|^k \cdot |A| \big)$ arcs, this can be done in polynomial time when $k$ is a fixed constant.
\end{proof}

Now, we continue with considering the \SupDirSPMinCost problem.

\begin{thm}\label{thm:sup_dac_constantk}
The \SupDirSPMinCost problem can be solved in polynomial time for directed acyclic graphs when $k$ is a constant.
\end{thm}
\begin{proof}
We construct a digraph $\widehat{D} =(\widehat{V}, \widehat{A})$ and a cost function $\widehat{c} \colon \widehat{A} \to \mathbb{R}$ as follows.
Let $I_{xy} = \big\{ i \in \{ 1, \ldots, k \} : xy \in A_i \big\}$.
The vertex set $\widehat{V}$ consists of all $k$ tuples $(v_1, \ldots, v_k)$ of the vertices of $D$.
For an arc $xy\in A$ and for a nonempty subset $J_{xy} \subseteq I_{xy}$, add an arc from $(u_1, \ldots, u_k)$ to $(v_1, \ldots, v_k)$ in $\widehat{D}$ if
\begin{enumerate}[label=(\alph*), topsep=1pt, itemsep=-2pt]
 \item $u_i = x$ and $v_i=y$ for all $i\in J_{xy}$, and
 \item $u_j = v_j$ for all $j\notin J_{xy}$,
\end{enumerate}
and let its $\widehat{c}$-cost be equal to $c(xy)$.
In this auxiliary digraph, the $(s,\ldots,s)$-$(t,\ldots,t)$ dipaths correspond to the feasible solutions of the \SupDirSPMinCost problem.
Analogously to the proof of Theorem~\ref{thm:exact_dac_constantk}, one can argue that the shortest $(s, \ldots, s)$-$(t, \ldots, t)$ dipaths in $\widehat{D}$ naturally correspond to the optimal solutions of \SupDirSPMinCost.
\end{proof}

We next discuss the fixed-parameter tractability of our problems.
That is, we design algorithms with running time bounded by $f(\ell) \cdot \mathrm{poly}(N)$, where $N$ is the size of the input, $\ell$ is a certain parameter, and $f \colon \Z_+ \to \Z_+$ is an arbitrary computable function~\cite[page~5]{Flum2006Parametrized}, where $\Z_+$ denotes the set of nonnegative integers.
Such algorithms are called FPT algorithms.

We first show that the \SupDirSPMinCost and \SupUndSPMinCost problem is fixed-parameter tractable parameterized by the number of multi-colored arcs or edges.

\begin{thm}\label{thm:intersectionFPT_SupDir}
  There exists an FPT algorithm for \SupDirSPMinCost and \SupUndSPMinCost parameterized by the number $\ell$ of multi-colored arcs or edges, respectively.
  The algorithm takes $k \cdot 2^\ell$ executions of a shortest path algorithm in some subgraphs of the input graph with nonnegative cost function.
\end{thm}
\begin{proof}
  We only present the proof for digraphs since the same proof works with undirected graphs as well.
  Let $F\subseteq A$ denote the set of multi-colored arcs.

  First we show that negative costs can be replaced with $0$.
  To see this, observe that an optimal solution to the original problem can be obtained by adding the originally negative cost arcs to an optimal solution with respect to this modified cost function.
  Hence, we may assume that the cost function $c$ is nonnegative.

  Let $A^* \subseteq A$ be an optimal solution to the \SupDirSPMinCost problem.
  We show that given $A^* \cap F$, one can extend it to an optimal solution in polynomial time.
  To see this, modify the costs of the arcs in $A^* \cap F$ to $0$, and observe that an optimal solution containing exactly $A^* \cap F$ from $F$ must contain a shortest $s$-$t$ dipath in the digraph $D_i=(V,A_i)$ with respect to the modified cost function.
  If we take the union of a shortest $s$-$t$ dipath in $D_i$ for each $i\in\{1,\ldots,k\}$ and the arcs in $A^* \cap F$, then we obtain an optimal solution to \SupDirSPMinCost.

  It follows from this observation that we can choose the subset $\widetilde{F} \subseteq F$ in every possible way, and modify the cost of the arcs in $\widetilde{F}$ to $0$, and then find a shortest dipath $P_i$ in $D_i$ for each $i\in\{1,\ldots,k\}$ with respect to this modified cost function.
  Then, we take the union of these dipaths, and choose the best (i.e.\ the one with minimum modified cost) among them.
  This gives an optimal solution to \SupDirSPMinCost with respect to $c$.
  This takes $k \cdot 2^\ell$ executions of a shortest path algorithm which we had to show.
\end{proof}

On the other hand, for the \ExUndSPMinCost problem, one can decide whether a feasible solution exists when the number of multi-colored edges is a parameter.

\begin{thm}\label{thm:ExUndSPEx_fpt}
There exists an FPT algorithm for \ExUndSPEx parameterized by $\ell$, where $\ell$ is the number of multi-colored edges.
The algorithm solves $k \cdot \ell! \cdot 4^\ell$ instances of the \textsc{DisjointPaths} problem with at most $\ell$ pairs of terminals in some subgraphs of the input graph.
\end{thm}
\begin{proof}
 Let $F\subseteq E$ denote the set of multi-colored edges.
 Let $E^* \subseteq E$ be a feasible solution to the \ExUndSPEx problem.
 We show that given $E^* \cap F$, one can extend it to a feasible solution by solving at most $k \cdot \ell! \cdot 2^\ell$ instances of \textsc{DisjointPaths} with at most $\ell$ pairs of terminals.
 Observe that $E^* \cap F \cap E_i$ is the union of at most $\ell$ disjoint paths for each $i \in \{ 1, \ldots, k \}$.
 For each $i \in \{ 1, \ldots, k \}$, we do the following.
 Take an order $P_{i,1}, \ldots, P_{i,\ell_i}$ of these paths and choose one of the endpoints $s_{i,j}$ of $P_{i,j}$ for each $j \in \{ 1, \ldots, \ell_i \}$.
 Let $t_{i,j}$ be the other endpoint of $P_{i,j}$.
 Now solve the \textsc{($\ell_i+1$)-DisjointPaths} problem with terminal pairs $(s,s_{i,1}), (t_{i,\ell_i},t)$, and $(t_{i,j}, s_{i,j+1})$ in the subgraph $(V, E_i \setminus F)$ for each $j \in \{ 1, \ldots, \ell_i-1 \}$.
 Then $E^* \cap F$ and the so-obtained paths for each $i \in \{ 1, \ldots, k \}$ together form a feasible solution of \ExUndSPEx.

 We can choose the subset $\widetilde{F} \subseteq F$ in every possible way such that $\widetilde{F} \cap E_i$ is the union of disjoint paths, and apply the procedure described above.
This way, we solve at most $k \cdot \ell! \cdot 4^\ell$ instances of \textsc{DisjointPaths} with at most $\ell$ pairs of terminals.
Clearly, \ExUndSPEx is feasible if and only if any of these instances of \textsc{DisjointPaths} is feasible.

 Since there exists an FPT algorithm for \textsc{DisjointPaths} parameterized by the number of terminal pairs~\cite{Kawarabayashi2012Disjoint}, we obtain an FPT algorithm for \ExUndSPEx parameterized by~$\ell$.
\end{proof}

Observe that Theorem~\ref{thm:ExUndSPEx_fpt} implies that \ExUndSPEx can be solved in polynomial time when $k=2$ and $|E_1\cap E_2|=1$, in contrast to the case of digraphs (cf.\ Theorem~\ref{thm:exact_directed_NPC}).

The algorithm presented in the proof of Theorem~\ref{thm:ExUndSPEx_fpt} can be easily applied to the \ExUndSPMinCost as follows.
In the modified approach, instead of finding arbitrary disjoint paths, we need to find those of minimum total cost, then their total costs must be increased by the cost of the chosen multi-colored edges $\widetilde{F}$, and the output is the cheapest among them.
Since no polynomial-time algorithm is known for the problem of finding disjoint paths of minimum total cost, we do not obtain an FPT algorithm for \ExUndSPMinCost in general.
However, some highly restricted cases of the problem of finding disjoint paths of minimum total cost can be solved in polynomial time, see, for example~\cite{Bjorklund2019Shortest, Borradaile2015Towards, ColinDeVerdiere2008Shortest, Hirai2018Shortest, Kobayashi2010Shortest}.

Now let us consider the case of digraphs.
The algorithm presented in the proof of Theorem~\ref{thm:ExUndSPEx_fpt} can be easily modified for \ExDirSPEx by letting $s_{i,j}$ be the starting point of the dipath $P_{i,j}$ for each $i,j$ (i.e., we follow the orientation of the dipaths instead of considering both their traversals).
Thus, we need to solve $k \cdot \ell! \cdot 2^{\ell}$ instances of the \textsc{DisjointDipaths} problem with at most $\ell$ pairs of terminals.
Although this problem is already NP-complete for two pairs of terminals~\cite{Fortune1980Directed}, there exists an FPT algorithm for planar digraphs parameterized by the number of desired dipaths~\cite{Cygan2013Planar}.
This algorithm combined with the approach described above gives an FPT algorithm for \ExDirSPEx in planar digraphs.

\begin{cor}\label{cor:ExDirSPEx_planar}
For planar digraphs, there exists an FPT algorithm for \ExDirSPEx parameterized by $\ell$, where $\ell$ is the number of multi-colored arcs.
\end{cor}

Finally, we consider a simple case that can be solved in polynomial time.
A set family on a finite ground set is called \emph{laminar} if, for any two of its intersecting sets, one is contained in the other.
A set family on a finite ground set is called a \emph{chain} if it is laminar and does not contain any disjoint sets.

\begin{thm}\label{thm:laminarfamily}
If the color classes form a laminar family, then \ExDirSPMinCost, \ExUndSPMinCost, \SupDirSPMinCost, and \SupUndSPMinCost are poly\-no\-mial-time solvable.
\end{thm}
\begin{proof}
  In the \ExDirSPMinCost and the \ExUndSPMinCost problems, there exists a feasible solution if and only if the laminar family is a disjoint union of chains and each inclusion-wise minimal member of the family contains a directed or undirected $s$-$t$ path.
  In this case, any optimal solution is the disjoint union of a shortest directed or undirected $s$-$t$ in each inclusion-wise minimal member of the family.

In the \SupDirSPMinCost and the \SupUndSPMinCost problems, there exists a feasible solution if and only if each inclusion-wise minimal member of the family contains a directed or undirected $s$-$t$ path.
In this case, any optimal solution contains all the negative-cost arcs or edges, an arbitrary subset of the zero-cost arcs or edges, and the arcs or edges of such a directed or undirected $s$-$t$ path in each inclusion-wise minimal member of the family which is shortest with respect to the cost function obtained from $c$ by replacing its negative values with 0.
\end{proof}

\section{Acknowledgement}

The authors are grateful to the organizers of the 12th Eml\'ekt\'abla Workshop which took place in G\'ardony, Hungary in 2022.
The research of Naonori Kakimura was supported by JSPS KAKENHI Grant Numbers 23K21646, JP22H05001, Japan and JST ERATO Grant Number JPMJER2301, Japan.
The work of P\'eter Madarasi was supported by the Ministry for Culture and Innovation and Technology of Hungary from the National Research, Development and Innovation Fund, financed under the ELTE TKP 2021-NKTA-62 funding scheme, and by the Ministry of Innovation and Technology NRDI Office within the framework of the Artificial Intelligence National Laboratory Program, and by the Lend\"ulet Programme of the Hungarian Academy of Sciences --- grant number LP2021-1/2021.
P\'eter Madarasi and Kitti Varga were supported by the Ministry of Innovation and Technology of Hungary from the National Research, Development and Innovation Fund --- grant number ADVANCED 150556.
Jannik Matuschke was supported by the special research fund of KU Leuven (project C14/22/026).

\bibliography{bibliography}
\bibliographystyle{plain}

\end{document}